\documentclass[cis]{ipart_v1}

\firstpage{1}
\Vol{XX}
\Issue{X}
\Year{2020}

\usepackage{tikz}
\usetikzlibrary{matrix,decorations.pathreplacing}

\usepackage{enumerate}
\usepackage{color}

\newcommand{\FA}{\mathbf{A}}   
\newcommand{\FB}{\mathbf{B}}   
\newcommand{\Fu}{\mathbf{u}}     
\newcommand{\Fx}{\mathbf{x}}     
\newcommand{\Fv}{\mathbf{v}}     

\newcommand{\Ff}{\mathbf{f}}     
\newcommand{\FN}{\mathbf{N}}     
\newcommand{\Ft}{\mathbf{t}}     
\newcommand{\Fs}{\mathbf{s}}     
\newcommand{\Fz}{\mathbf{z}}     

\newcommand{\BI}{\mathbb{I}}     
\newcommand{\BT}{\mathbb{T}}
\newcommand{\BM}{\mathbb{M}}

\newcommand{\BA}{\mathbb{A}}
\newcommand{\BB}{\mathbb{B}}
\newcommand{\BR}{\mathds{R}}  

\newcommand{\Chi}{\mathds{1}}

\newcommand{\ES}{\mathcal{W}} 

\newcommand*\TRANS{{\mathpalette\doTRANS\empty}}
\makeatletter
\newcommand*\doTRANS[2]{\raisebox{\depth}{$\m@th#1\intercal$}}
\makeatother
\newcommand\MATRIX[1]{\begin{bmatrix} #1 \end{bmatrix}}

\usepackage{dsfont}
\usepackage[font=footnotesize]{subfig}

\theoremstyle{plain}
\newtheorem{theorem}{Theorem}

\newtheorem{proposition}{Proposition}
\newtheorem{lemma}{Lemma}

\theoremstyle{plain}
\newtheorem{assumption}{Assumption}

\theoremstyle{plain}
\newtheorem{definition}{Definition}

\theoremstyle{definition}
\newtheorem{remark}{Remark}

\theoremstyle{plain}
\newtheorem{strategy}{Strategy}

\begin{document}

\title[Linear Quadratic Graphon Field Games]{Linear Quadratic Graphon Field Games}

\author[S.~Gao, R.~ Foguen~Tchuendom and P.E.~Caines]{Shuang Gao$^\ast$, Rinel~Foguen~Tchuendom$^\ast$ and Peter~E.~Caines$^\ast$\blfootnote{In honor of Professor Tyrone Duncan on the occasion of his 80th birthday. \\
\indent This work is supported  in part by NSERC Grant 2019-05336 (Canada), U.S. ARL and ARO Grant W911NF1910110, and USAF Grant  FA9550-19-1-0138.\\
\indent The authors would like to thank Rabih Salhab and Shujun Liu for helpful suggestions and discussions.}}

\begin{abstract}
 Linear quadratic graphon field games (LQ-GFGs)  are defined to be linear quadratic games  which involve a large number of agents that are weakly coupled via a weighted undirected graph  on which each node represents an agent. The links of the graph correspond to couplings between the agents' dynamics, as well as between  the  individual cost functions, which each agent attempts to minimize. We formulate limit LQ-GFG problems based on the assumption that these graphs lie in a sequence which converges to a limit graphon. First, under a finite-rank assumption on the limit graphon, the existence and uniqueness of solutions to the formulated limit LQ-GFG problem is established. Second, based upon the solutions to the limit LQ-GFG problem,   $\varepsilon$-Nash equilibria  are constructed for the corresponding game problems with a very large but finite number of players. 
This result is then generalized to the case with random initial conditions. It is to be noted that LQ-GFG problems are distinct  from the class of graphon mean field game (GMFG) problems where  a population is hypothesized to be associated with each node of the graph  \cite{PeterMinyiCDC18GMFG,PeterMinyiCDC19GMFG}.
\end{abstract}

\maketitle

\section{Introduction}

Strategic decision problems over very large-scale networks arise in applications such as 5G communication, large-scale social networks, stock market networks, advertising networks, electrical networks and so on. However decision and control problems for such systems require tractable solutions that are of low computational complexity.

 When networks are complete and uniform, the couplings between agents appear in their dynamics and performance functions as mean field terms. In those cases where the population is large, Mean Field Game theory (\cite{HMC06,HCM07,lasry2006jeux1,lasry2006jeux2}) may then be applied  in order to analyse the possible Nash equilibria of the overall system. On the other hand, for a large class of non-uniform networks progress has been made in various directions. Such work includes, for example, mean field games with localities \cite{huang2010nce},  quantilized mean field games (\cite{crisan2014conditional, Tembine, FTR-RM-PC}) and  {Graphon Mean Field Game} theory \cite{PeterMinyiCDC18GMFG,PeterMinyiCDC19GMFG}.

Graphon theory  provides an important framework for the study of very large graphs, convergent sequences of dense graphs and for the construction and analysis of their limit objects (\cite{borgs2008convergent,borgs2012convergent,lovasz2012large}).
 The theory has been used in the analysis of dynamical models such as the heat equation and  coupled oscillators  (\cite{medvedev2014nonlinear,medvedev2014nonlinear2,chiba2019mean}), network centrality \cite{avella2018centrality},  random walks over large dense graphs \cite{petit2019random}, the 
Graphon Control  of dynamical systems coupled over very large-scale networks (\cite{ShuangPeterCDC17,ShuangPeterTAC18,ShuangPeterCDC18,ShuangPeterCDC19W2,ShuangPeterCDC19W1}), and 
static  and dynamic game problems on graphons \cite{parise2018graphon,carmona2019stochastic,PeterMinyiCDC18GMFG,PeterMinyiCDC19GMFG}.

The recently developed Graphon Control theory \cite{ShuangPeterCDC17,ShuangPeterTAC18,ShuangPeterCDC18,ShuangPeterCDC19W2,ShuangPeterCDC19W1}  employs the graphon model  to represent control systems on arbitrary-sized networks. This enables the study of control problems for very large-scale network-coupled dynamical systems and generates low-complexity approximate control solutions to  otherwise intractable problems. The solutions are either centralized solutions \cite{ShuangPeterTAC18} or collaborative solutions \cite{ShuangPeterCDC19W2}. This current work studies the approximate solutions in a  competitive situation.

 {Graphon Mean Field Game} (GMFG) theory was proposed and developed in \cite{PeterMinyiCDC18GMFG,PeterMinyiCDC19GMFG}
wherein a large number of weakly coupled competitive agents are distributed over a large non-uniform graph, and consequently each agent is associated with a nodal mean field.
Within this framework network wide Nash equilibria and $\varepsilon$-Nash results have been established in both  non-linear and linear quadratic cases.   

Mean field games on networks have  appeared in \cite{huang2010nce,gueant2015existence, delarue2017mean}. In \cite{gueant2015existence} the graph is the state space of the mean-field game problem representing physical constraints on the state space. 
While in \cite{delarue2017mean} linear-quadratic mean-field games over Erd\"os-R\'enyi graphs are studied where the associated asymptotic game is a classical mean field game.
These formulations are different from the current work
in  their assumptions concerning their finite and asymptotic features.
We also note that, similar to \cite{delarue2017mean}, in the current work each node represents an agent and this is different from \cite{PeterMinyiCDC18GMFG, PeterMinyiCDC19GMFG} where each node is associated with a population of agents.
However, it is worth mentioning that when the underlying graphons in the current paper are taken to be step function graphons, the problems on networks with nodal populations can be equivalently formulated. 
This work is related to the work in \cite{huang2010nce} where mean-field game problems with non-homogeneous dense network  weightings in the running costs are studied. But it differs from  \cite{huang2010nce} in both the problem formulation and the solution method. A very recent work \cite{lacker2020case} studies mean field games on sparse graphs by exploring transitive properties in graph structures. 

In the current work, we explicitly solve a class of LQ-GFGs with deterministic dynamics and with random initial conditions by exploiting the spectral decomposition of the underlying graphon limit. Furthermore, the corresponding $\varepsilon$-Nash results are established.

\subsection*{Notation}
$\BR$ represents the space of all real numbers. 
$\|\cdot\|_\infty$ denotes the standard infinity norm for matrices and vectors, that is, for any matrix $A \in \BR^{n\times n}$, $\|A\|_\infty \triangleq \max_i\sum_j|a_{ij}|$ and for any vector $v \in \BR^n$, $\|v\|_\infty \triangleq \max_i |v_i|$.  $L^2[0,1]$ denotes the standard Lebesgue space over $[0,1]\subset \BR$ under the $\|\cdot\|_2$-norm defined by $\|\Fv\|_2=(\int_0^1 \Fv(\alpha)^2 d\alpha)^{1/2}$.
 {For any $P \subset [0,1]$, $\mathds{1}_{P}\in L^2[0,1]$ represents the piece-wise constant function with $1$ in $P$ and $0$ elsewhere.}
We use the upper bound big $O$ notation in this paper, that is, for two functions $f$ and $g$ defined on some subsets of real number, $f=O(g)$ means that there exists a positive real number $M$ and a number $x_0$ such that $|f(x)|\leq M g(x)$ for all $x\geq x_0$.

\section{Graphons and Graphon Dynamical Systems}\label{sec:graphon-dynamical-system}

 Graphs can be considered as models for network couplings.  A graph
 $G=(V,E)$ is specified by a node set $V=\{1,...,N\}$ and an edge set $E\subset V\times V$. The corresponding adjacency matrix $A=[a_{ij}]$ is defined as follows: $a_{ij}=1$ if $(i,j) \in E$ otherwise $a_{ij} =0$. A graph is undirected if its edge pair is unordered. 
 For a weighted undirected graph, $a_{ij}$ in its adjacency matrix is given by the weight  between nodes $i$ and $j$. 
 Furthermore an adjacency matrix can be represented as a pixel diagram on the unit square $[0,1]^2 \subset \BR^2$, which corresponds to a graphon step function \cite{lovasz2012large}.

 Graphons are formally defined as symmetric Lebesgue measurable functions $\FA: [0,1]^2\rightarrow [0,1]$. The space of graphons endowed with the \emph{cut metric} \cite{lovasz2012large} allows us to define the convergence of graph sequences. 
 In this paper, we consider symmetric Lebesgue measurable functions $\FA:[0,1]^2 \rightarrow [-c, c]$ with $c>0$, the space of which is denoted by $\ES_c$. The space $\ES_c$ is compact under the cut metric after identifying points of cut distance zero  \cite{lovasz2012large}.
 %
%
 \begin{figure}[htb]
    \centering
    \includegraphics[height=1.8cm]{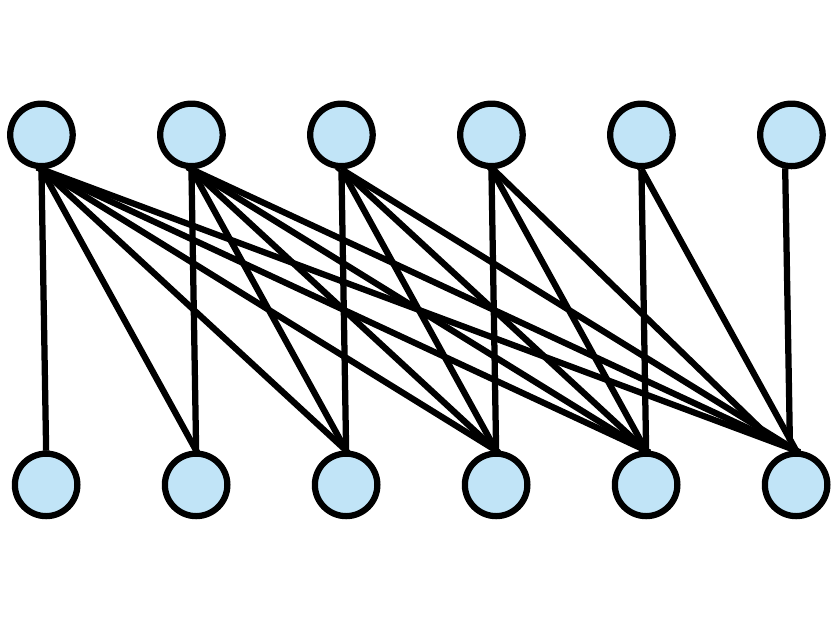}\qquad
    \includegraphics[height=1.8cm]{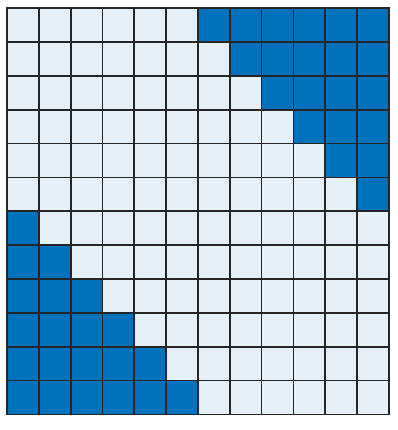}\qquad
    \includegraphics[height =1.8cm]{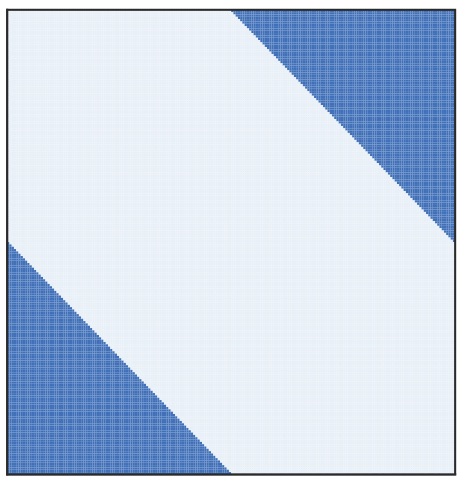}
    \caption{{Half graphs and its limit graphon \cite{lovasz2012large}}}
    \label{fig: converge-in-pixel-pictures}
\end{figure}
A graphon $\FA \in \ES_c$ defines a self-adjoint bounded linear operator from $L^2[0,1]$ to $L^2[0,1]$ as follows:
$$[\FA \Fv](\gamma) = \int_{[0,1]} 
 \FA(\gamma, \eta)\Fv(\eta)d\eta, \quad \forall \gamma \in [0,1],$$
 where $\Fv,  \FA \Fv \in L^2[0,1]$. 

Let 
$L^2([0,T];L^2[0,1])$  
denote the Banach space of equivalence classes of strongly measurable (in the B\"ochner sense \cite[p.103]{showalter2013monotone}) mappings $[0,T] \rightarrow L^2[0,1]$ that are integrable with norm
$$
\begin{aligned}
&\| \Ff \|_{L^2([0,T];L^2[0,1])} =\left(\int_0^T \|\Ff(s)\|^2_2 ds\right)^{{1/2}}
\end{aligned}
$$
The space of continuous functions from $[0,T]$ to a Hilbert space $\mathcal{H}$ is denoted by $C([0,T];\mathcal{H})$.

An infinite dimensional time-dependent graphon linear control system $(\BA_t; \BB_t)$ is formulated as follows: 
\begin{equation} \label{equ: infinite-system-model}
  \begin{aligned}
  &\dot{\Fx}_\Ft= \BA_t \Fx_\Ft + \BB_t \Fu_\Ft,  \quad \mathbf{x_0} \in L^2[0,1],~ t\in[0,T],   \\
  \end{aligned}
\end{equation}
where $\BA_t = \alpha_t\BI +\FA, \BB_t= \beta_t \BI + \FB$ with $\alpha_{(\cdot)}, \beta_{(\cdot)} \in C([0,T]; \BR)$,  $\FA,\FB \in \ES_c$,  $\mathbf{x_t} \in L^2[0,1]$ is the system state at time $t$, and $\mathbf{u}_{(\cdot)} \in L^2([0,T]; L^2[0,1])$ is the control input function over time.
Verifying all the conditions in \cite[Part II, Proposition 3.4, 3.6]{bensoussan2007representation}, we obtain that the system \eqref{equ: infinite-system-model} is well defined and has the unique mild (and strong) solution in $C([0,T];L^2[0,1])$ 
given by 
\begin{equation}
    \Fx_\Ft = \Phi(t,0)\Fx_0 + \int_0^t \Phi(t,\tau) \BB_\tau \Fu_\tau d\tau
\end{equation}
where the evolution operator is given by $$\Phi(t,\tau)\triangleq \text{exp}\left(\int_\tau^t \alpha_s ds \right) \text{exp}((t-\tau)\FA).$$

{The time-varying graphon dynamical system formulation in \eqref{equ: infinite-system-model} will be employed to support the well-posedness of the forward backward equations in \eqref{eq:fixed-point} in Section \ref{sec:solution-deterministic}. }

Readers are referred to \cite{lovasz2012large} for a thorough exposition of graphon theory and to \cite{ShuangPeterTAC18} for a definitive exposition of graphon control systems and their relation to finite network control systems.

\section{Deterministic Linear Quadratic Graphon Field Games}
In this section we introduce the deterministic linear quadratic graphon field game model on a weighted undirected graph,  demonstrate the applications of the spectral decomposition to solve the limit problem, and then establish the $\varepsilon$-Nash property for finite graphon field game problems.

\subsection{Finite Population and Finite Graph Problems}
Consider the following {(time-invariant)} linear  quadratic graphon field game problem on a weighted undirected graph with the dynamics for the $i$th agent given by
\begin{equation}
  \begin{aligned}
  \label{finite-game-dynamics}
   &\dot{x}^i_t= \alpha x_t^i+\beta u_t^i + \eta\frac1N\sum_{j=1}^{N}a_{ij}x^j_t,\\
  & t \in [0,T],~~ \alpha, \beta \in \BR, ~~  x_t^i, u_t^i \in \BR, ~~  i\in\{1,...,N\}, 
  \end{aligned}
 \end{equation} 
 where   $A_N\triangleq[a_{ij}]$ represents the adjacency matrix of the underlying weighted undirected graph and $\{x_0^i\}_{i=1}^N$ are  initial conditions. The  objective of the $i${th} agent is to minimize its cost given by
\begin{equation}
\label{finite-game-cost}
  J^i(u^i, u^{-i}) = \frac12 \int_0^T \Big[\big(x^i_t - \frac1N\sum_{j=1}^{N}a_{ij}x^j_t\big)^2 + r(u_t^i)^2 \Big]dt, \quad 
  \text{with}~ r>0,
\end{equation}
where $u^i \in \mathcal{U} \triangleq L^2([0,T];\BR)$ for all $i\in\{1,...,N\}$.

We next define the \emph{$N$-uniform partition} $\{P_1, \ldots,P_N\}$ of $[0,1]$ as $P_1= [0,\frac1N]$ and $P_k=(\frac{k-1}{N},\frac{k}{N}]$ for $2\leq k \leq N$.  
The step function graphon $\FA^\FN \in \ES_c$ that corresponds to $A_N$ is given by 
\begin{equation}
    \FA^\FN(\vartheta,\varphi) = \sum_{i=1}^{N} \sum_{j=1}^{N} \Chi_{_{P_i}}(\vartheta)\Chi_{_{P_j}}(\varphi)a_{ij},  \quad (\vartheta,\varphi) \in [0,1]^2.
\end{equation}
%
Let the piece-wise constant function $\Fx_t^\FN \in L^2{[0,1]}$ corresponding to $x_t \triangleq [x_0^1, x_0^2, \ldots, x_0^N]^\TRANS \in \BR^N$ be given by 
\begin{equation}\label{eq:pwc}
	\Fx_t^\FN (\vartheta) =\sum_{i=1}^N {\Chi}_{_{P_i}}(\vartheta) x_t^i, \quad \forall \vartheta \in [0,1].
\end{equation}
Similarly, $\Fu^\FN\triangleq \{\Fu^\FN_{t}, t\in [0,T]\}$ can be defined.
Based on the construction procedure, each agent on the finite network is associated with a partition element in the $N$-uniform partition of $[0,1]$.

{
\begin{definition}
Define $\Fz_t^\gamma \triangleq \int_{[0,1]}\FA(\gamma, \alpha) \Fx_t(\alpha) d\alpha$ as the Local Graphon Field for agent $\gamma \in [0,1]$ of a graphon dynamical system at time $t\in[0,T]$, where $\FA$ is the underlying graphon and $(\Fx_t^\gamma)_{\gamma \in [0,1]}$ is the state of the system at time $t \in [0,T]$.
	 The collection of local graphon fields $(\Fz^\gamma_t)_{ \gamma \in[0,1]}$ is then defined as the {Graphon Field} of the system at time $t\in[0,T]$. 
\end{definition}
}


Clearly in an $N$-agent problem, the graphon field is given by the piece-wise constant function $\Fz_t^\FN$ that corresponds to the $N$-dimensional vector $z_t \triangleq \frac1N A_N x_t$ following \eqref{eq:pwc}, for all $t\in [0,T]$, and the local graphon field  affecting the agent indexed by $i$ is given by $\Fz^{\FN \gamma}_t = \int_{[0,1]}\FA^\FN(
\gamma, \alpha) \Fx_t^\FN (\alpha)d\alpha$, with  $\gamma\in P_i \subset [0,1]$ and ${t\in [0,T]}$.

\subsection{Solutions to the Limit Problems}
\label{sec:solution-deterministic}

 Letting the network cardinality go to  infinity, the limiting game may be formulated. Since in the limit the effect of an individual agent on the graphon field becomes negligible, the resulting  minimization problems may be treated as independent linear quadratic tracking problems.

For the limit problem to be well defined, we need the following assumption. 
\begin{assumption}
 \label{ass:limit-assumptions}
~ 
There exist $\FA \in \ES_c$ and $\Fx_0 \in L^2[0,1]$ such that
\[
\textup{(a)} ~~ \lim_{N\rightarrow\infty}\|\FA^\FN -\FA\|_\textup{op} =0, \qquad \quad \textup{(b)}~~ \lim_{N\rightarrow\infty} \|\Fx_0^\FN - \Fx_0\|_2 =0,
\]
where $\|\cdot\|_{\textup{op}}$ denotes the operator norm.
\end{assumption}
The convergence for the underlying graphons can also be defined in the cut norm $\|\cdot\|_{\Box}$ (see e.g. \cite{lovasz2012large}) since the following relation \cite{janson2010graphons,parise2018graphon} holds: 
\begin{equation}
    \|\FA\|_{\Box} \leq \|\FA\|_\textup{op} \leq \sqrt{8\|\FA\|_{\Box}}, \quad \forall \FA \in \ES_1.
\end{equation}
%
%

The dynamics in equation \eqref{finite-game-dynamics} for all agents can be represented by an infinite dimensional system equation as in \eqref{equ: infinite-system-model} where the graph is represented by the corresponding step function graphon and the state and the control are respectively represented by piece-wise constant functions $\Fx^\FN_t$ and $\Fu^\FN_t$ in  $L^2[0,1]$  (see \cite{ShuangPeterTAC18}). Under Assumption \ref{ass:limit-assumptions} and the assumption that $\Fu^\FN$ converges to a limit control $\Fu$ in $C([0,T]; L^2[0,1])$ under the uniform norm, this representation permits a well-defined limit equation  where the convergence of trajectories is in the space $C([0,T]; L^2[0,1])$ under the uniform norm following a slight extension of the result in  \cite[Theorem 7]{ShuangPeterTAC18}.  
Then for almost all $\gamma \in [0,1]$, we can and shall write the evolution of the $\gamma $ component system as in \eqref{limit-game-dynamics} below where  the family of local mean fields satisfies \eqref{gmfg-consistency-condition} below.

 A \emph{Nash equilibrium} for the infinite population LQ Graphon Field Game associated with a limit of the system \eqref{finite-game-dynamics} and individual performance functions \eqref{finite-game-cost}  is characterized as follows:
\begin{enumerate}[(\bf C1)]
\item \emph{Best Response}\\
For a given local graphon field trajectory $(\Fz_t^\gamma )_{t\in [0,T]}$, let the best response  $\Fu^{\gamma,*}$  for the  agent indexed by $\gamma  \in[0,1]$  be given 
by the solution to the linear tracking problem with the controlled linear dynamics
\begin{equation}
\label{limit-game-dynamics}
  \begin{aligned}
   &\dot{\Fx}^\gamma_t= \alpha \Fx_t^\gamma+ \beta \Fu_t^\gamma + \eta \Fz^\gamma_t  ,\quad  \Fx_0^\gamma= x_0^\gamma, \\
  & t \in [0,T],\quad \alpha, \beta \in \BR,  \quad \Fx^\gamma_t, \Fu^\gamma_t \in \BR,
  \end{aligned}
 \end{equation} 
for almost all $\gamma \in [0,1]$, with the performance  function
\begin{equation}
\label{limit-game-cost}
  J(\Fu^\gamma, \Fz^{\gamma}) =\frac12 \int_0^T \Big((\Fx^\gamma_t - \Fz^\gamma_t)^2 + r(\Fu_t^\gamma)^2 \Big) dt,
\end{equation}
where  
$\Fu^{\gamma,*} = \arg \inf_{\Fu^\gamma \in \mathcal{U}} J(\Fu^{\gamma}, \Fz^{\gamma}).$
%
\item \emph{Consistency or Equilibrium Condition} \\
Given state trajectories $( \Fx^{\gamma,*,\Fz}_t)_{t\in [0,T]}$ under the  best response control $\Fu^{\gamma, *}$,  the local graphon field trajectories $(\Fz^\gamma_t)_{t\in [0,T]}$   satisfy the consistency conditions: 
\begin{equation}
\label{gmfg-consistency-condition}
  \forall t \in [0,T], \quad   \Fz_t^\gamma = [\FA \Fx^{*,\Fz}_t]^\gamma
\end{equation}
for almost every $\gamma \in  [0,1]$ where $\Fx^{*,\Fz}_t \triangleq ( \Fx^{\gamma,*,\Fz}_t)_{\gamma \in [0,1]}$. 

\end{enumerate}

 We note that a related notion of a local mean field appears in \cite{PeterMinyiCDC18GMFG, PeterMinyiCDC19GMFG}.


  A procedure to find an equilibrium solution to the infinite population infinite graphon problem presented in \textbf{(C1)} and \textbf{(C2)} above is as follows. 
\subsubsection*{Best Response}

The best response solution  to the above linear quadratic tracking problem  is given by  (see e.g. \cite{lewis2012optimal}) 
  \begin{align}
      \Fu_t^\gamma &= -\frac{\beta}{r} \pi_t \Fx_t^\gamma+ \frac{\beta}{r} \Fs_t^\gamma,\label{eq:control}\\
      -\dot{\pi}_t &= 2\alpha \pi_t -\frac{\beta^2}{r}\pi^2_t + 1, \quad \pi_T = 0, \label{eq:P-evo}\\
      -\dot{\Fs}_t^\gamma & = \Big(\alpha - \frac{\beta^2}{r} \pi_t \Big) \Fs_t^\gamma +  (1- \eta \pi_t)\Fz_t^\gamma,\quad  \Fs_T^\gamma = 0, \label{eq:s-evo}
  \end{align}
  for almost all $\gamma \in [0,1]$.
The stars in the notations   are being dropped henceforth. 
%
 Consequently the closed loop state equation for  the agent indexed by $\gamma$ is given by
\begin{equation}
        \dot{\Fx}_t^\gamma = \Big( \alpha - \frac{\beta^2}{r} \pi_t \Big)\Fx^\gamma_t  + \eta \Fz^\gamma_t + \frac{\beta^2}{r} \Fs^\gamma_t , \hspace{2mm} \Fx_0^\gamma = x_0^\gamma, \label{eq:x-evo} \\
\end{equation}
  for almost all $\gamma \in [0,1]$.

\subsubsection*{Consistency or Equilibrium Condition}

In the space  $L^2([0,T],L^2[0,1])$,  we search for the \emph{Graphon Field} $\Fz$ and the off-set term $\Fs$ that ensure the consistency condition holds.
We now invoke {Assumption \ref{ass:limit-assumptions}} and observe that an application of $\FA$ on each side of \eqref{eq:x-evo} yields  that the consistency condition \eqref{gmfg-consistency-condition} is equivalent to the existence of a unique solution to 
the following infinite dimensional ordinary differential equations over $[0,T]$,
\begin{equation}\label{eq:fixed-point}
 \begin{aligned}
    \dot{\Fz}_t &= \Big(\alpha-\frac{\beta^2}{r}\pi_t\Big)\Fz_t+ \eta \FA \Fz_t +  \frac{\beta^2}{r} \FA \Fs_t, \quad  \Fz_0= \FA \Fx_0,\\
    -\dot{\Fs}_t &=\Big(\alpha-\frac{\beta^2}{r}\pi_t\Big)\Fs_t + (1-\eta \pi_t)\Fz_t, \quad \Fs_T = 0.
 \end{aligned} 
\end{equation}
Each of the two equations above given the solution to the other is well defined following {the well-posedness of \eqref{equ: infinite-system-model}. }
%
 Assumption \ref{ass:limit-assumptions} ensures that the solution to each of the equations in \eqref{eq:fixed-point} is the limit of solutions  in the  $L^2([0,T],L^2[0,1])$ sense to the sequences of equations corresponding to the finite network problems, when all solutions exist.  Based on Section \ref{sec:graphon-dynamical-system}, if the solutions $\Fz$ and $\Fs$ exist in $L^2([0,T];L^2[0,1])$, they must also lie in $C([0,T]; L^2[0,1])$.

Next, we derive the sufficient conditions for the existence and uniqueness of the solutions to the joint forward backward problem in \eqref{eq:fixed-point}.   
\begin{assumption}\label{ass:spectral-graphon}
  The graphon limit admits a finite spectral representation $\FA = \sum_{\ell=1}^d  \lambda_\ell \Ff_\ell \Ff_\ell^\TRANS$, where $\{\lambda_\ell\}_{\ell=1}^d$ is the set of all the non-zero eigenvalues and $\{\Ff_\ell\}_{\ell=1}^d$ is corresponding set of orthonormal eigenfunctions.
\end{assumption}
{The graphon $\FA(x,y)=1$ for all $x,y \in [0,1]$ is a trivial low-rank graphon example which corresponds to the standard uniform mean field coupling, and its rank is just $1$. In general, as an operator any graphon is compact and eigenvalues accumulate at zero \cite{lovasz2012large}. Thus the above assumption corresponds to a reasonable approximation. }

 Under this assumption,  
$ \Fz_t^\gamma= [\FA \Fx_t]^\gamma  = \sum_{\ell=1}^{d}\lambda_\ell \langle\Ff_\ell , \Fx_t \rangle\Ff_\ell(\gamma),$
in the $L^2[0,1]$ sense and from \eqref{eq:fixed-point}
in the eigendirection $\Ff_\ell$, $\ell \in \{1,...,d\}$, we obtain
\begin{equation}\label{eq:fixed-point-eigen}
  \begin{aligned}
    \dot{z}^\ell_t &= \Big(\alpha-\frac{\beta^2}{r}\pi_t + \eta \lambda_\ell \Big)z^\ell_t + \frac{\beta^2}{r} \lambda_\ell s^\ell_t,  \quad 
     z_0^\ell = \langle \Fz_0, \Ff_\ell\rangle = \lambda_\ell \langle \Fx_0, \Ff_\ell\rangle ,\\
    \dot{s}^\ell_t&=-\Big(\alpha-\frac{\beta^2}{r}\pi_t\Big)s^\ell_t - (1-\eta \pi_t)z^\ell_t,  \quad  s^\ell_T = 0,
  \end{aligned}
\end{equation}
over the interval $[0,T]$, where $z^\ell_t = \langle \Fz_t, \Ff_\ell\rangle$ and  $ s^\ell_t = \langle \Fs_t, \Ff_\ell\rangle$.

From \eqref{eq:fixed-point}, it may be verified that the projections of $\Fz_t$ and $\Fs_t$ into the subspace $\mathds{S}^\perp$ (that is, the complementary subspace orthogonal to $\mathds{S}\triangleq\text{span}\{\Ff_1,...,\Ff_d\}$) are zero for all $t\in[0,T]$. Therefore, for all $t\in[0,T]$,
\begin{equation}
  \Fz_t = \sum_{\ell=1}^d z_t^\ell \Ff_\ell \qquad \text{and} \qquad \Fs_t = \sum_{\ell=1}^d s_t^\ell \Ff_\ell .
\end{equation}

Following \cite{huang2012social,bensoussan2016linear,salhab2016collective}, we associate the solvability of problems \eqref{eq:fixed-point-eigen} to the solvability of Riccati equations.

\begin{assumption}\label{ass:Riccati-Sol-Existence}
For any $\lambda_\ell, ~\ell  \in \{1,\ldots,d\}$, there exists a solution to the Riccati equation
\begin{align}\label{equ:Riccati-Equation} 
    - \dot{\Pi}_t^\ell &=   \Big[2\big(\alpha-\frac{\beta^2}{r}\pi_t\big)+\eta \lambda_\ell\Big] \Pi_t^\ell + \frac{\beta^2}{r} \lambda_\ell (\Pi_t^\ell)^2 +(1-\eta \pi_t), 
    ~\Pi_T^\ell &= 0 ,  
\end{align}
over the interval $[0,T]$, where  $\pi_{(\cdot)}$ is the solution to the Riccati equation in \eqref{eq:P-evo}.
\end{assumption}
Note that finite escape time may appear for the solutions to the Riccati equation above  depending on the parameters and the time horizon, in particular when  $\lambda_\ell >0$.  

Under {Assumption \ref{ass:Riccati-Sol-Existence}}, for each $\lambda_\ell$, the solution to the Riccati equation \eqref{equ:Riccati-Equation} is unique due to the smoothness of the right-hand side with respect to $\Pi^\ell_t$ (see \cite{salhab2016collective}).
Let $q_t^\ell\triangleq \Pi_t^\ell z^\ell_t$.   Then
\begin{equation}  
  \frac{d\big({q}_t^\ell - {s}_t^\ell\big)}{dt} =  - \left[\Big(\alpha-\frac{\beta^2}{r}\pi_t\Big) +\frac{\beta^2}{r}\lambda_\ell\Pi_t^\ell \right](q_t^\ell- {s}_t^\ell)
\end{equation}
with the terminal condition $(q^\ell_T - s^\ell_T) =0$. Solving this ordinary differential equation (ODE) allows us to conclude that $s_t^\ell\triangleq\Pi_t^\ell z_t^\ell$ for all $t\in[0,T]$ (see also \cite{salhab2016collective}). 
Replacing $s_t^\ell$ in the forward equation of \eqref{eq:fixed-point-eigen} by $\Pi_t^\ell z_t^\ell$, we obtain
\begin{equation}
      \dot{z}^\ell_t = \left[\alpha + \frac{\beta^2}{r} (\Pi_t^\ell \lambda_\ell - \pi_t) + \eta \lambda_\ell \right]z^\ell_t,  \quad  z_0^\ell  = \lambda_\ell \langle \Fx_0, \Ff_\ell\rangle.\\
\end{equation}

An alternative approach to establish the sufficient condition for the existence of a unique solution to the problem \eqref{eq:fixed-point} is given below where the counterpart of {Assumption \ref{ass:Riccati-Sol-Existence}} is a contraction condition in \eqref{eq:all-contraction}. The relation between the two sufficient conditions shall be analyzed in future work.

\begin{proposition}[Appendix \ref{sec:PropositionProof}]
\label{prop:one}
  Under {Assumption \ref{ass:spectral-graphon}}, the two-point boundary value problem \eqref{eq:fixed-point} has a unique solution if the following condition holds for all $\ell \in \{1,\ldots,d\}$: 
  \begin{equation}\label{eq:all-contraction}
  \int_0^T \frac{\beta^2}{r}\frac{ |\lambda_\ell|}{B_\ell(\tau)} \int_\tau^T |(1-\eta \pi_s)|B_\ell(s)ds  d\tau <1
\end{equation}
where $B_\ell(t)= \textup{exp}\left[{\int_0^t\Big(\alpha-\frac{\beta^2}{r}\pi_\tau+\eta\lambda_\ell\Big)d\tau-\int_t^T\Big(\alpha-\frac{\beta^2}{r}\pi_\tau\Big)d\tau}\right]$
 and $\pi_{(\cdot)}$ is the solution to the Riccati equation in \eqref{eq:P-evo}. 
\end{proposition}

 The eigenvalues represent the amplitude of the network influences, which relate directly to the sufficient condition \eqref{eq:all-contraction} for the existence of a unique solution to the fixed-point equation \eqref{eq:fixed-point}.

\begin{theorem}[\textbf{Finite-Rank LQ-GFG Equations for Limit Problems}]
{
Under {Assumptions \ref{ass:limit-assumptions}, \ref{ass:spectral-graphon}} $\&$ {\ref{ass:Riccati-Sol-Existence}}, the equilibrium solution to the limit graphon field game problem is explicitly given by 
\begin{equation}\label{equ:mfg-solution}
      \Fu_t^\gamma = -\frac{\beta}{r} \pi_t \Fx_t^\gamma+ \frac{\beta}{r} \sum_{\ell=1}^d s_t^\ell \Ff_\ell(\gamma), \quad \text{for almost all}~ \gamma \in [0,1],\\
\end{equation}
where for all $\ell \in\{1,...,d\}$, $t \in [0,T]$, 
\begin{align}
s_t^\ell &= \Pi_t^\ell z_t^\ell, \label{eq:s-ell}\\
    -\dot{\pi}_t &= 2\alpha \pi_t -\frac{\beta^2}{r}\pi^2_t + 1, \quad \pi_T = 0,\label{eq:riccati-1}\\
    - \dot{\Pi}_t^\ell &=  \Big[2\big(\alpha-\frac{\beta^2}{r}\pi_t\big)-\eta \lambda_\ell\Big] \Pi_t^\ell + \frac{\beta^2}{r} \lambda_\ell (\Pi_t^\ell)^2  +(1-\eta \pi_t), \quad \Pi_T^\ell  = 0,\label{eq:riccati-2}  \\
    \dot{z}^\ell_t &= \Big[\alpha + \frac{\beta^2}{r} (\Pi_t^\ell \lambda_\ell - \pi_t) +\eta \lambda_\ell \Big]z^\ell_t,\quad z_0^\ell  = \lambda_\ell \langle \Fx_0, \Ff_\ell\rangle.  \label{eq:eigen-mean-state}
\end{align}
}
\end{theorem}

Sequentially solving \eqref{eq:riccati-1}, \eqref{eq:riccati-2} and \eqref{eq:eigen-mean-state} yields the offest term as in \eqref{eq:s-ell}. Thus this procedure provides an explicit hierarchical  decoupling (from $\pi_{(\cdot)}$ to $\Pi^\ell$ to $z^\ell$) of the joint equations in \eqref{eq:fixed-point-eigen}. 

\begin{remark}
The initial condition for $z_0^\ell$ depends on the labeling of the network, since it is given by $z_0^\ell  = \lambda_\ell \langle \Fx_0, \Ff_\ell\rangle.$  This means that $s^\ell$ and hence the best response  depend on the labeling. 
Therefore, the labeling should be fixed in the first step to generate the best response law. 
\end{remark}
\begin{remark}
Although the rank of the underlying graphon limit is assumed to be finite, the limit graphon field game problem still involves an infinite number of agents.  

It it worth mentioning that  any finite graph can be represented by a step function graphon (which is a special case of finite-rank graphons)  and hence any finite agent problem  can be reformulated in an infinite dimensional space based on graphons and $L^2[0,1]$ functions. However, the exact solution cannot be given by the Finite-Rank Graphon Field Game Equations  \eqref{equ:mfg-solution}, \eqref{eq:s-ell},  \eqref{eq:riccati-1}, \eqref{eq:riccati-2} and \eqref{eq:eigen-mean-state} (with a simple replacement of the graphon limit by a step function), since in this case each individual is no longer negligible to the evolution of the graphon field. One needs to differentiate a limit graphon which happens to be a step function and a finite network step function based on the number of agents in the game problem.
\end{remark}

\subsection{$\varepsilon$-Nash Property for Finite  Problems}
In this section, $\varepsilon$-Nash equilibrium is constructed from the limit LQ-GFG solution for the corresponding large (but finite) population games.

\begin{definition} An $N$-tuple of strategies $(u^{1},..., u^{N})$ generates an $\varepsilon$-Nash equilibrium ($\varepsilon>0$) if  the following holds
 $$J({u}^{i} , u^{- i}) \leq \inf_{v^i\in \mathcal{U}}J(v^i, u^{- i})+ \varepsilon$$ 
 for each $i \in \{1,...,N\}$,
 where $\mathcal{U} \triangleq L^2([0,T], \BR)$,  $J({u}^{ i},  u^{- i})$ denotes the cost for agent $i$ when all the agents employ the strategies in $(u^{ 1},..., u^{ N})$, and $J(v^i, u^{- i})$ denotes the cost for agent $i$ when it deviates unilaterally by taking response law $v^i$.
\end{definition}

Let $\Fx_0^\FN$ be the piece-wise constant function with the $N$-uniform partition of $[0,1]$ corresponding to $x_0 = [x_0^1, \ldots,x_0^N]^\TRANS$. 

\begin{strategy}[Finite Problem Strategies: Deterministic Case] \label{strategy-deterministic}
 Let the $N$-tuple $({u}^{o 1},..., u^{oN})$ of strategies be
  constructed as follows:
%
for any agent $i \in\{1,...,N\}$,
\begin{equation}
\label{eq:Nash-Control-Case12}
\begin{aligned}
  {u}^{oi}_t 
 &=- \frac{\beta}{r} \pi_t x_t^{oi}+ \frac{\beta}{r} \bar \Fs_t^i \\
   \bar \Fs_t^i & \triangleq  \frac1{\mu{(P_i)}}\int_{P_i}  \Fs^\gamma d\gamma = \sum_{\ell=1}^d s_t^\ell \frac1{\mu{(P_i)}}\int_{P_i}  \Ff_\ell(\gamma) d\gamma,
    \end{aligned}
    \end{equation} 
 where $\pi_{(\cdot)} $ and $\{s_t^\ell\}_{\ell=1}^d$ are generated from the limit LQ-GFG solutions \eqref{eq:s-ell}, \eqref{eq:riccati-1}, \eqref{eq:riccati-2} and \eqref{eq:eigen-mean-state},  the initial conditions for \eqref{eq:eigen-mean-state} are given by $z^\ell_0 = \lambda_0\langle \Fx_0^{\FN}, \Ff_\ell\rangle$, $\ell \in \{1,\ldots,d\}$, $x_t^{oi}$ denotes the state of agent $i$ at time $t$, and $\mu{(P_i)}$ denotes the size of $P_i$ (which is $1/N$ for the case with the $N$-uniform partition).
 \end{strategy}

We now present sufficient conditions under which the $N$-tuple of strategies $( {u}^{o i})_{i= 1}^N$ indeed generates an $\varepsilon$-Nash equilibrium, for the corresponding large (but finite) population games.

\begin{assumption} \label{ass:a_ii=0}
For all $i, j \in \{1,2,\ldots\}$,  $a_{ii} =0$, and there exists $c> 0$ such that $|a_{ij}|\leq c$.

\end{assumption}

Let $\FA^\FN$ be the corresponding step function of the $N\times N$ adjacency matrix $A_N=[a_{ij}]$ of the underlying graph, $1\leq i,j\leq N.$  
\begin{assumption}\label{ass:net-convegence-revised}
The sequence $\{\FA^\FN\}_{N=1}^\infty$ and its limit graphon $\FA$ satisfy
\begin{equation}
  \max_{i\in 1,\ldots,N} \frac{1}{\mu(P_i)} \left\|(\FA-\FA^\FN)\mathds{1}_{P_i} \right\|_2 \rightarrow 0,\quad \text{ as } N\rightarrow \infty,
\end{equation}
where $\{P_1,\ldots P_N\}$ forms an $N$-uniform partition of $[0,1]$.
\end{assumption} 
 We now present  the $\varepsilon$-Nash property  for the finite problem which is established in Appendix \ref{sec: proof-deterministic-LQ-GFG} following the procedure in \cite{huang2005nash}. 

\begin{theorem}[Appendix \ref{sec: proof-deterministic-LQ-GFG}]
\label{thm:deterministic-LQ-GFG} 
Under {Assumptions \ref{ass:limit-assumptions}, \ref{ass:spectral-graphon}, \ref{ass:Riccati-Sol-Existence} 
$\&$  \ref{ass:a_ii=0}}, 
the following holds for any agent $i \in \{1,..., N\}$ 
\begin{equation}
\begin{aligned}
  &J(u^{oi}, u^{-oi}) - \inf_{u^i\in \mathcal{U}}J(u^{i}, u^{-oi}) 
  =  \max  \left\{O(E_N),O(E_N^2)\right\}
\end{aligned}
\end{equation}
where $$E_N\triangleq\max_{i\in \{1,\ldots,N\}} \frac{1}{\mu(P_i)} \left\|(\FA-\FA^\FN)\mathds{1}_{P_i} \right\|_2;$$
furthermore, if {Assumption \ref{ass:net-convegence-revised}} also holds, then for any $\varepsilon>0$ there exists $N_0$ such that for any $N> N_0$ the following holds
\begin{equation}
  J(u^{oi}, u^{-oi}) \leq  \inf_{u^\in \mathcal{U}}J(u^{i}, u^{-oi}) + \varepsilon,
\end{equation}
that is $(u^{o1},..., u^{oN})$ generates an $\varepsilon$-Nash equilibrium for $N>N_0$.
\end{theorem}

\section{Random Initial Conditions}

An  $N$-agent game problem with network interactions is formulated as follows:
\begin{equation}
  \begin{aligned}
  \label{eq:finite-game-dynamics-stochastic-init}
   &\dot{x}^i_t= \alpha x_t^i+\beta u_t^i + \eta\frac1N\sum_{j=1}^{N}a_{ij}x^j_t, ~ t \in [0,T],~ \alpha, \beta \in \BR, ~ i \in \{1,...,N\} 
  \end{aligned}
 \end{equation} 
 where $x_0^i \sim \text{N}(\mu, \sigma^2)$ 
and $\{x_0^i\}_{i=1}^N$ are independent. 
 %
 The objective of the $i${th} agent  is the minimization of the performance function given by 
\begin{equation}
\label{eq:finite-game-cost-stochastic-init}
  J^i(u^i, u^{-i}) =\frac12 \mathds{E}\int_0^T \Big[\big(x^i_t - \frac1N\sum_{j=1}^{N}a_{ij}x^j_t\big)^2 + r(u_t^i)^2 \Big]dt
\end{equation}
where $r>0$ and $[a_{ij}]$ is the adjacency matrix of the underlying weighted undirected graph.

\subsection{Limit Function for Initial Conditions}

\begin{definition}
Consider a sequence of Gaussian random variables $\{x_0^1, x_0^2, ...\}$. Let $\Fx_0^\FN$ be the stochastic {piece-wise constant} function corresponding to the vector $(x_0^1, x_0^2, \ldots, x_0^N)^\TRANS$.
 A function $\Fx_0 \in L^2[0,1]$ is the limit function of the sequence $\{\Fx_0^\FN\}_{N=1}^\infty$ if 
$$
\forall \Fv \in L^2[0,1], \quad \langle \Fx_0, \Fv\rangle = \lim_{N\rightarrow \infty} \langle \Fx_0^\FN, \Fv\rangle
$$
in the mean square sense,
which we denote by  $\Fx_0 \triangleq \lim_{N\rightarrow \infty} \Fx_0^\FN.$
\end{definition}

For any basis function $\Ff_\ell$ in an orthonormal base system for $L^2[0,1]$, we obtain
$$
\lim_{N\rightarrow \infty}\langle \Fx_0^\FN, \Ff_\ell\rangle =  \lim_{N\rightarrow \infty} \frac1{N} \sum_{i=1}^N \bar v_\ell(i)x_0^i
$$
where 
$$
\begin{aligned}
\bar v_\ell & \triangleq \Big[\frac{1}{\mu(P_1)}\langle \mathds{1}_{P_1}, \Ff_\ell\rangle, ..., \frac{1}{\mu(P_N)}\langle \mathds{1}_{P_N}, \Ff_\ell\rangle\Big]^\TRANS\\
& = \Big[\frac{1}{\mu(P_1)}\langle \mathds{1}_{P_1}, \bar \Ff_\ell\rangle, ..., \frac{1}{\mu(P_N)}\langle \mathds{1}_{P_N}, \bar \Ff_\ell\rangle\Big]^\TRANS
\end{aligned}
$$  
with $\bar \Ff_\ell$ as the stepfunction approximation of $\Ff_\ell$ based on $N$-uniform partitions of $[0,1]$.
By the contraction property in \cite[Proposition 3]{ShuangPeterTAC18}, we obtain
$
\|\bar \Ff_\ell\|_2\leq  \| \Ff\|_{2} = 1.
$
Therefore
$\|\bar v_\ell\|_2 \triangleq \sqrt{\bar v_\ell^\TRANS \bar v_\ell}= \sqrt{N} \|\bar \Ff_\ell\|_2 \leq \sqrt{N}.$

Let 
$
S_N \triangleq  \langle \Fx_0^\FN, \Ff_\ell\rangle = \langle \Fx_0^\FN, \bar \Ff_\ell\rangle=  \frac1{N} \sum_{i=1}^N \bar v_\ell(i)x_0^i .
$
Clearly, $$S_N \sim \text{N}\left( \mu \frac1{N} \sum_{i=1}^N \bar v_\ell(i) ,  \sigma^2\frac{1}{N^2} \sum_{i=1}^N (\bar v_\ell(i))^2\right).$$
 The expectation satisfies
\begin{equation}
    \lim_{N\rightarrow \infty} \mathds{E} \big[ S_N \big]  = \lim_{N\rightarrow \infty} \frac1{N} \sum_{i=1}^N \bar f_\ell(i) \mu = \langle \mu \mathds{1},\Ff\rangle
\end{equation}
and the variance  satisfies $\text{var}(S_N) = O\Big(1/N\Big)$.
Therefore,
$
\lim_{N\rightarrow \infty}\langle \Fx_0^\FN, \Ff_\ell\rangle =  \langle \mu \mathds{1}, \Ff_\ell\rangle, 
$
in the mean square sense, that is, 
$
\lim_{N\rightarrow \infty} \Fx_0^\FN   \triangleq \mu \mathds{1}.
$ 

A natural choice for the limit of the initial condition is $\mu \mathds{1}$. 
Recall that $ \bar \Ff_\ell$ is the approximation of $\Ff_\ell$ based on $N$-uniform partition of $[0,1]$ with $N$ as the size of the finite population game.  Hence it is obvious that $ \langle \mu \mathds{1},  \Ff_\ell\rangle= \langle \mu \mathds{1}, \bar \Ff_\ell\rangle$. 
\subsection{The $\varepsilon$-Nash Property}
Based on the solution to the limit LQ-GFG problem with random initial conditions, the following strategy can be constructed. 
\begin{strategy}[Finite Problem Strategies: Random Initial Conditions]\label{str:randinit}

Let the $N$-tuple $({u}^{o 1},..., u^{oN})$ of strategies be  
  constructed as follows:
   \begin{equation}\label{eq:rdinit-Nash-Control-Case12}
      \begin{aligned}
    {u}^{oi}_t &=- \frac{\beta}{r} \pi_t x_t^{oi}+ \frac{\beta}{r} \bar \Fs^i_t\\
       \bar \Fs_t^i & \triangleq  \frac1{\mu{(P_i)}}\int_{P_i}  \Fs^\gamma d\gamma = \sum_{\ell=1}^d s_t^\ell \frac1{\mu{(P_i)}}\int_{P_i}  \Ff_\ell(\gamma) d\gamma,
    \end{aligned}
    \end{equation}
 %
where $s_t^\ell = \Pi_t^\ell z_t^\ell$,
\begin{align}
    -\dot{\pi}_t &= 2\alpha \pi_t -\frac{\beta^2}{r}\pi^2_t + 1, \quad \pi_T = 0, \quad  t\in[0,T],\label{eq:rdinit-riccati-1}\\
    - \dot{\Pi}_t^\ell &=  \left[2\Big(\alpha-\frac{\beta^2}{r}\pi_t\Big)-\eta \lambda_\ell\right] \Pi_t^\ell + \frac{\beta^2}{r} \lambda_\ell (\Pi_t^\ell)^2  +(1-\eta \pi_t), \quad
    \Pi_T^\ell  = 0, \label{eq:rdinit-riccati-2} \\
    \dot{z}^\ell_t &= \left[\alpha + \frac{\beta^2}{r} (\Pi_t^\ell \lambda_\ell - \pi_t) +\eta \lambda_\ell \right]z^\ell_t,  \quad z_0^\ell  = \lambda_\ell \langle \mu \mathds{1},  \Ff_\ell\rangle,\label{eq:rdinit-eigen-mean-state}
\end{align}
and $\mu{(P_i)}$ denotes the size of $P_i$ (which is $1/N$ for the case with the $N$-uniform partition).

%


 \end{strategy}
 
 Compared to the deterministic case, the only difference in the strategies is the choice of the initial conditions. 
In this case, each agent only needs to take into account the expectation of the initial conditions in \eqref{eq:rdinit-eigen-mean-state} to compute the offset terms $\{s^\ell\}_{\ell=1}^d$. 
Thus this provides a decentralized solution.

\begin{theorem}[Appendix \ref{sec:proof-rinit}]
\label{thm:stochasticinit-LQ-GFG}
Under {Assumptions  \ref{ass:limit-assumptions}(a), \ref{ass:spectral-graphon},  \ref{ass:Riccati-Sol-Existence} 
$\&$ \ref{ass:a_ii=0}},  the following holds for any agent $i$ 
\begin{equation}
\begin{aligned}
  &J(u^{oi}, u^{-oi}) - \inf_{u_i\in \mathcal{U}} J(u^{i}, u^{-oi}) 
  = \max\left\{O\Big(\frac{1}{\sqrt{N}} \Big), O\Big(E_N \Big), O\Big(E_N^2 \Big)\right\} 
\end{aligned}
\end{equation}
based on  \textup{\bf Strategy \ref{str:randinit}},
where $$E_N\triangleq\max_{i\in \{1,\ldots,N\}} \frac{1}{\mu(P_i)} \left\|(\FA-\FA^\FN)\mathds{1}_{P_i} \right\|_2;$$
furthermore, if {Assumption \ref{ass:net-convegence-revised}} also holds, then for any $\varepsilon>0$ there exists $N_0$ such that for any $N> N_0$ the following holds
\begin{equation}
  J(u^{oi}, u^{-oi}) \leq  \inf_{u^\in \mathcal{U}}J(u^{i}, u^{-oi}) + \varepsilon,
\end{equation}
that is $(u^{o1},..., u^{oN})$ generates an $\varepsilon$-Nash equilibrium for $N>N_0$.
\end{theorem}


\section{Numerical Examples}

\begin{figure}[htb] 
\centering
    \subfloat[State]{\includegraphics[width=4cm,trim = {0.1cm 0 0.1cm 0}, clip]{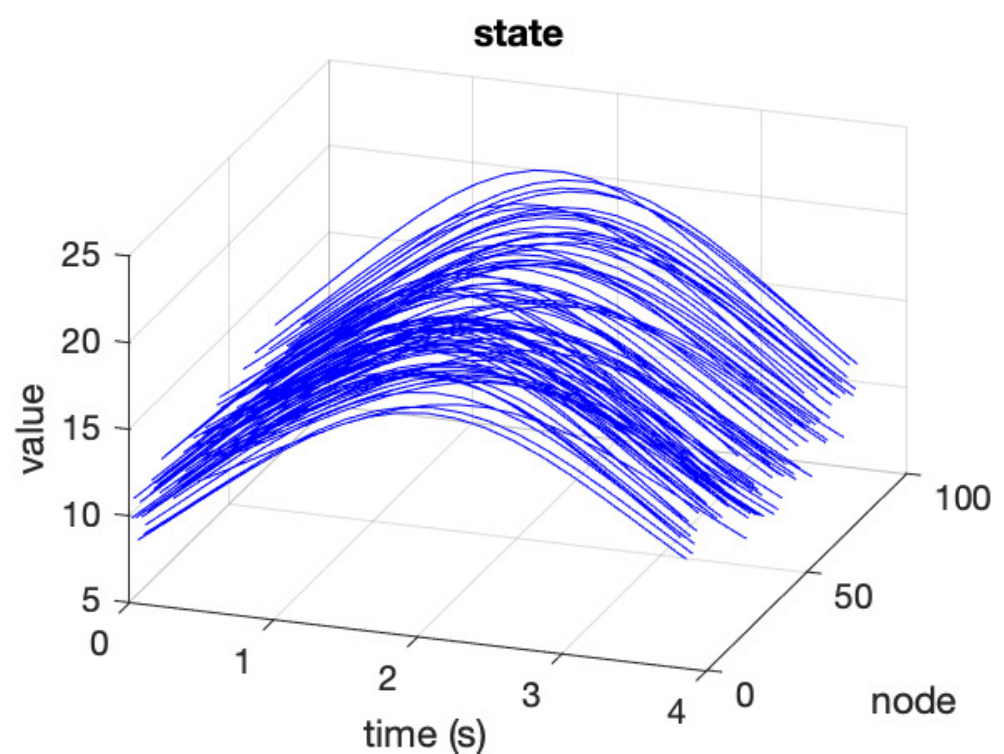}}~~
    \subfloat[GFG best response]{\includegraphics[width=4cm,trim = {0.1cm 0 0.1cm 0}, clip]{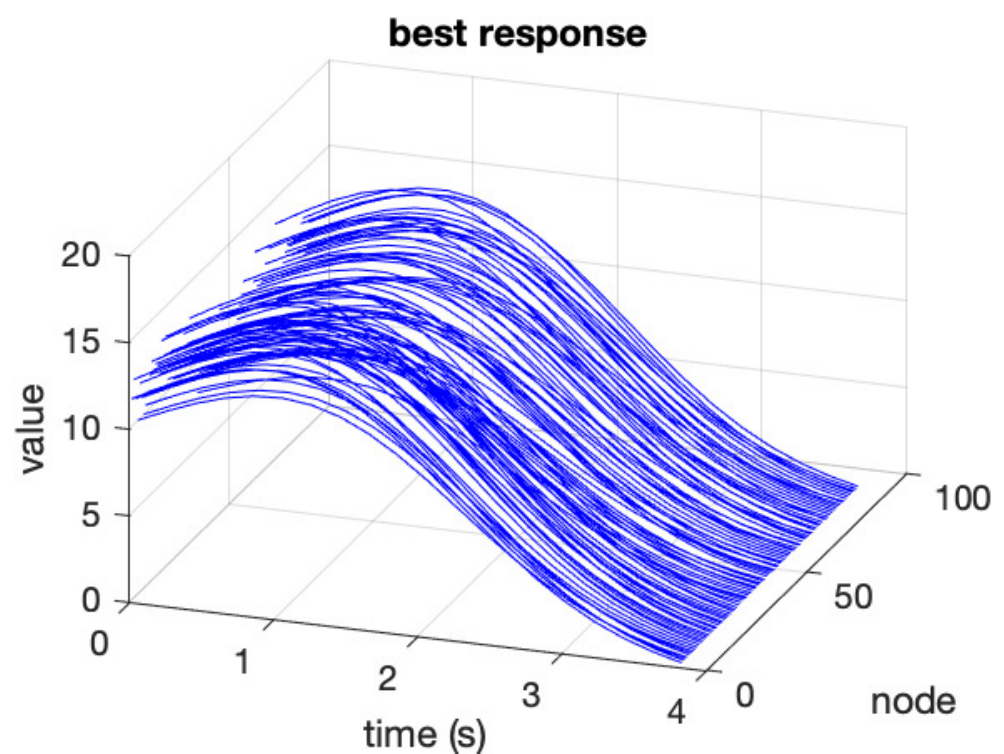}}~~
    \subfloat[Offset process $s$]{\includegraphics[width=4cm,trim = {0.1cm 0 0.1cm 0}, clip]{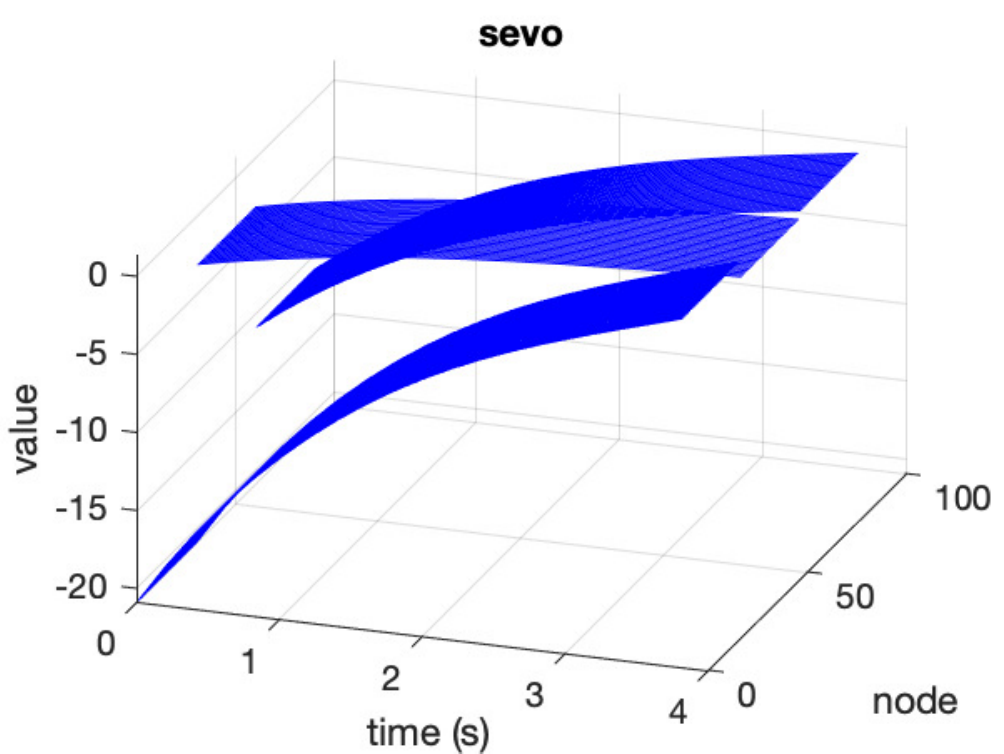}}\\
    \subfloat[Graphon field $z$]{\includegraphics[width=4cm,trim = {0.1cm 0 0.1cm 0}, clip]{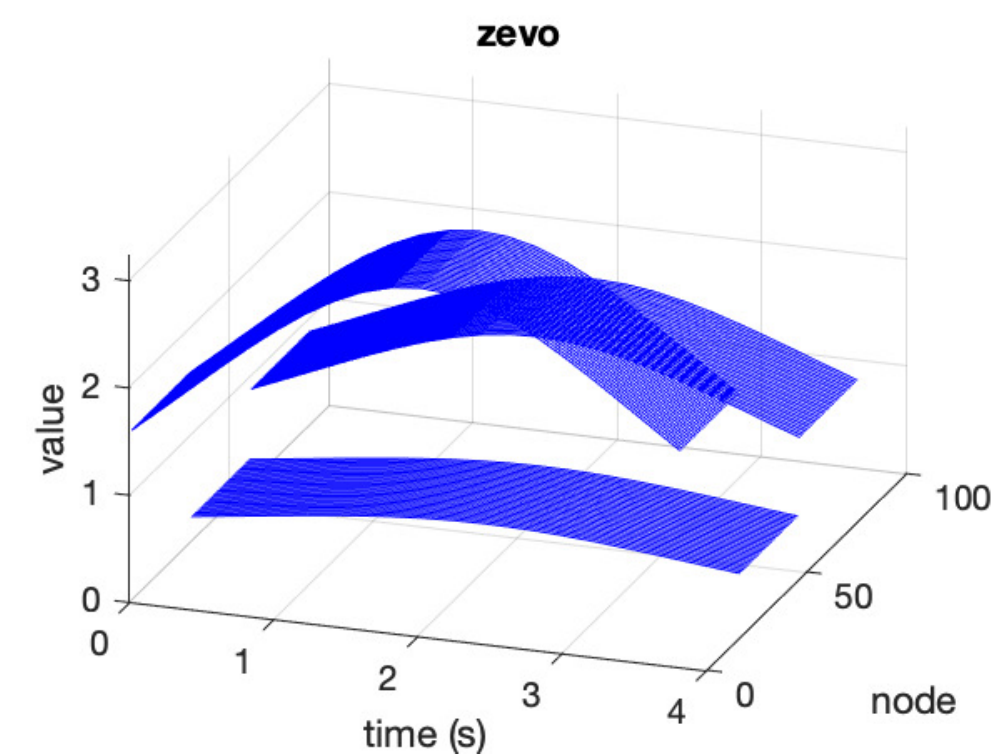}}~~
    \subfloat[Empirical graphon field]{\includegraphics[width=4cm,trim = {0.1cm 0 0.1cm 0}, clip]{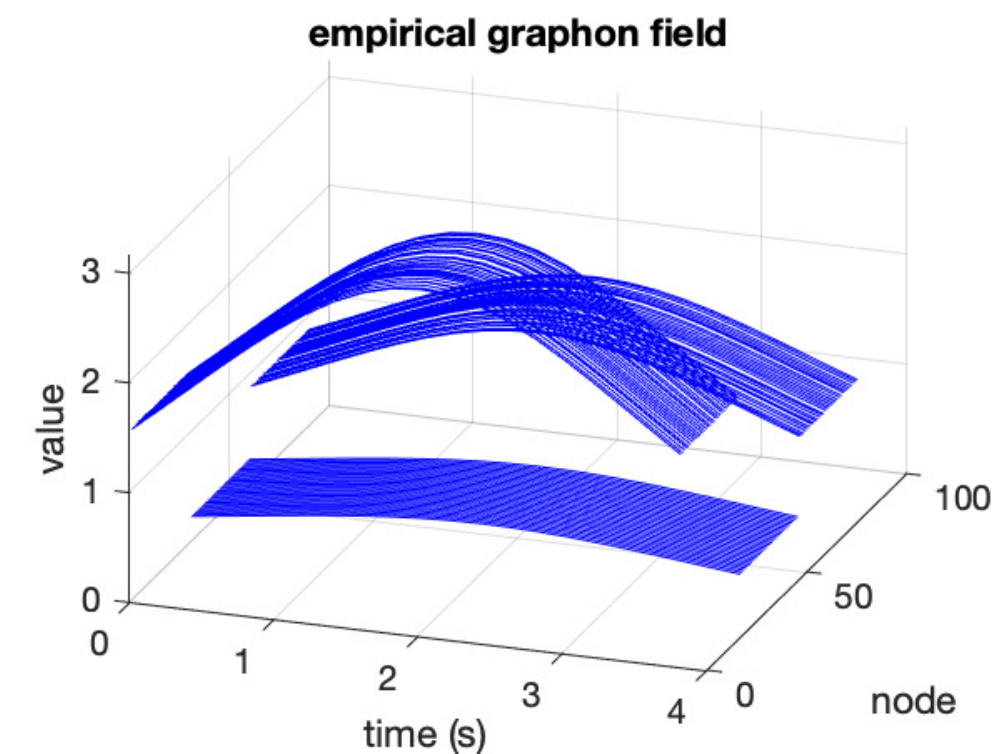}}~~
    \subfloat[Graphon field estimate error]{\includegraphics[width=4cm,trim = {0.1cm 0 0.1cm 0}, clip]{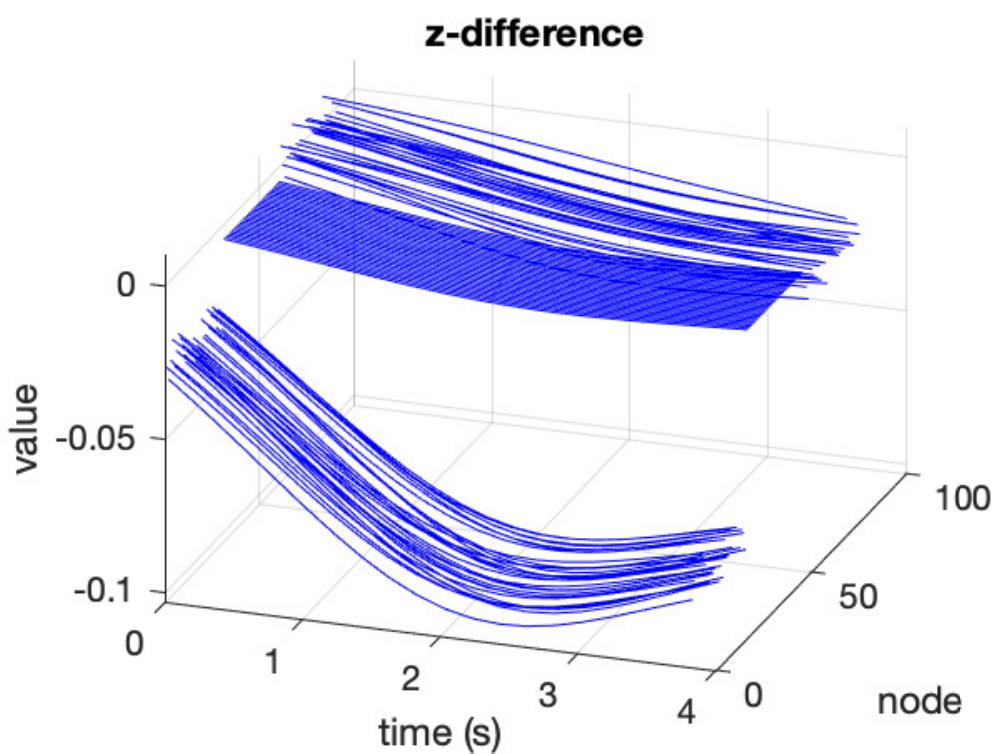}}\\
     \subfloat[Graph structure and its pixel representation]{\includegraphics[width=8.5cm,trim = {1cm 0 1cm 0}, clip]{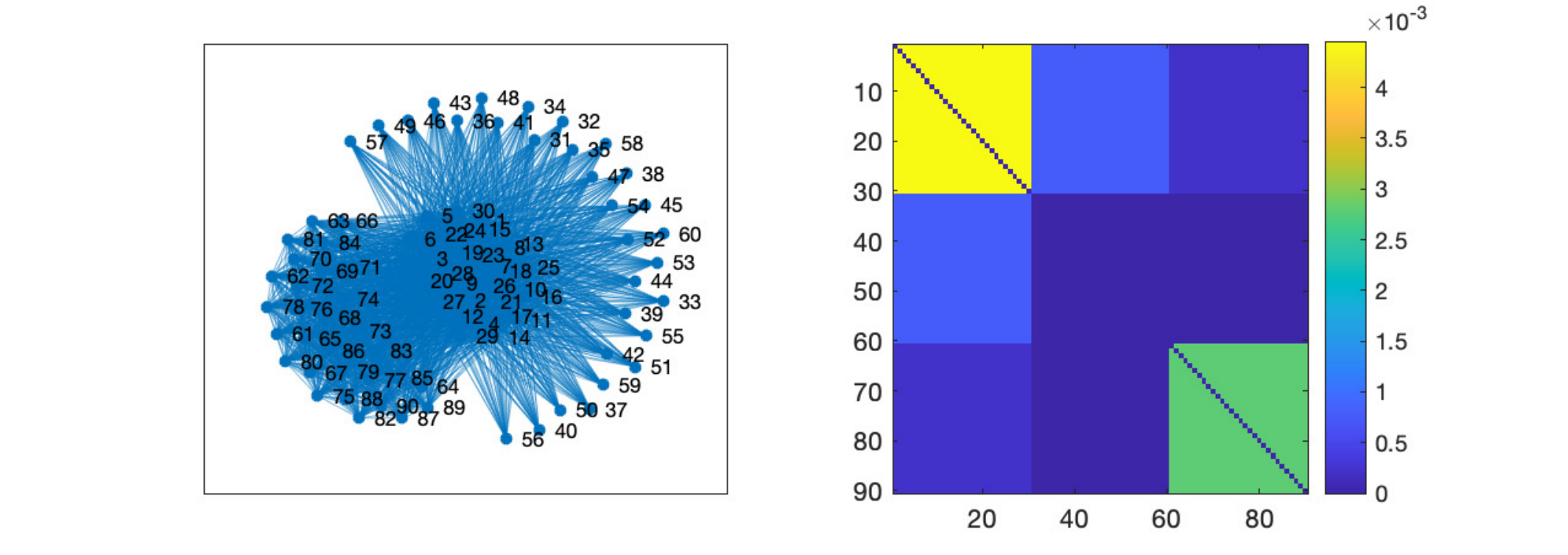}} ~~
    \subfloat[Individual cost]{\includegraphics[width=4cm,trim = {0cm 0 0cm 0}, clip]{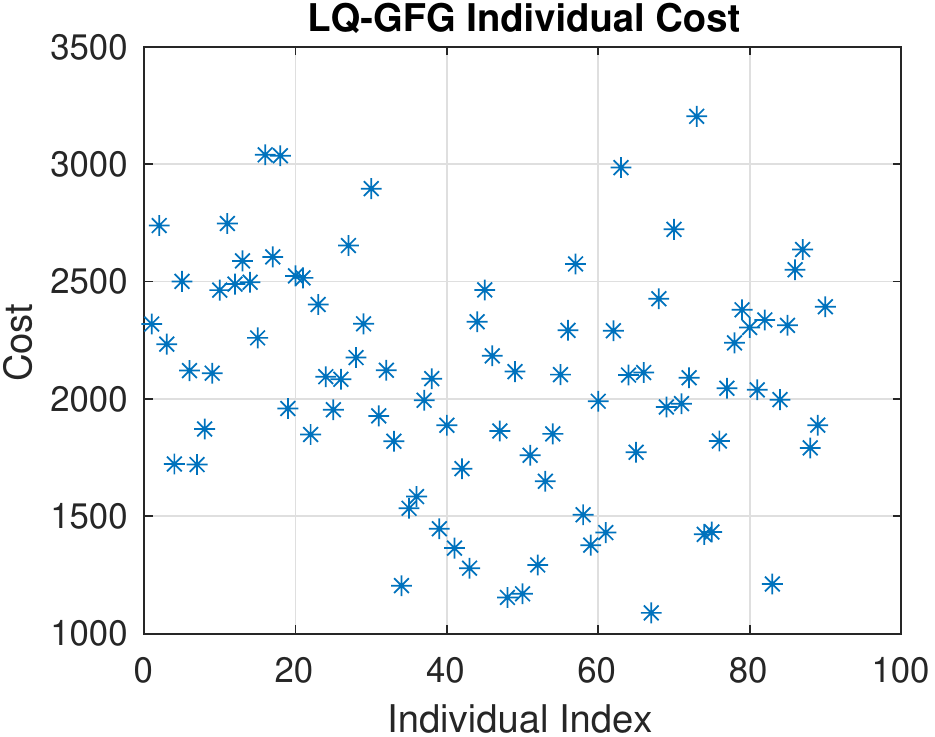}}
         \caption{Graphon field game simulation for systems on weighted multipartite graphs of size 90 with weights among and within the communities given by \eqref{eq:multiP}, where each community contains 30 nodes. We choose a problem with 90 nodes in this example for the convenience of illustrating the network structure, but the method could deal with problems of much large sizes. }\label{fig:fig-MultiP}
\end{figure}

The following values will be adopted for the parameters: 
 $\alpha = -0.5$, $\beta =1$, $\eta= 0.1$, $r =10$, $T=4$. The initial conditions $\{x_0^1,x_0^{2}...,\}$ are chosen as independent Gaussian random variables with variance $1$ and mean $10$. 

 \subsection{Example on Multipartite Graphs}
  We consider multipartite graphs (with no self-loops) where the connection weights are specified by the following matrix
\begin{equation}\label{eq:multiP}
    \MATRIX{0.25& 0& 0.02\\ 0& 0& 0.07\\ 0.02 &  0.07 & 0.40}.
\end{equation}
The sizes of the sequence of graphs are given by  $3n$ where $n \in \{1,2,...\}$ is the number of nodes in each community. 
Since  $n$ may vary, the underlying graph could be of arbitrary size. Clearly Assumptions \ref{ass:limit-assumptions}(a), \ref{ass:a_ii=0} and \ref{ass:net-convegence-revised} are satisfied. Since the  graphon limit has rank $3$, Assumption \ref{ass:spectral-graphon} holds.  Furthermore, the specific values for the parameters in the example allow Assumption \ref{ass:Riccati-Sol-Existence} to hold.  Hence, the result in Theorem \ref{thm:stochasticinit-LQ-GFG} applies. 
A simulation result on a network of size $90$ is shown in Figure \ref{fig:fig-MultiP}.

\subsection{Example on Graphs Generated from a Sinusoidal Graphon}
To generate a graph of size $N$ from the graphon limit $\FA(x,y) = 0.5\cos\pi(x-y)+0.5$ with $x, y \in [0,1]$, we first get the uniform grid in $[0,1]$ with grid points $p_1,...,p_N$ and then connect $i$ and $j$ ($i\neq j$) with weight $\FA(p_i,p_j)$. Clearly this generation procedure ensures Assumptions \ref{ass:limit-assumptions}(a), \ref{ass:a_ii=0} and \ref{ass:net-convegence-revised} are satisfied.

{The normalized $L^2[0,1]$ eigenvectors of the graphon limit $\FA(x,y) = 0.5\cos\pi(x-y)+0.5$ associated with nonzero eigenvalues are  $\mathds{1}_{[0,1]}$, $\sqrt{2}\sin\pi( \cdot)$ and $\sqrt{2}\cos\pi( \cdot)$, and the corresponding nonzero eigenvalue are  $\frac12$, $\frac14$ and $\frac14$ (see e.g., \cite{ShuangPeterCDC19W1}). }
 Hence the  graphon limit $\FA$ has rank $3$ and Assumption \ref{ass:spectral-graphon} holds.  Furthermore, the specific values for the parameters in the example allow Assumption \ref{ass:Riccati-Sol-Existence} to hold.  Hence, the result in Theorem \ref{thm:stochasticinit-LQ-GFG} applies. A simulation result is shown in Figure \ref{fig:fig-S3Equal}.

\begin{figure}[htb] 
\centering
    \subfloat[State]{\includegraphics[width=4cm,trim = {0.1cm 0 0.1cm 0}, clip]{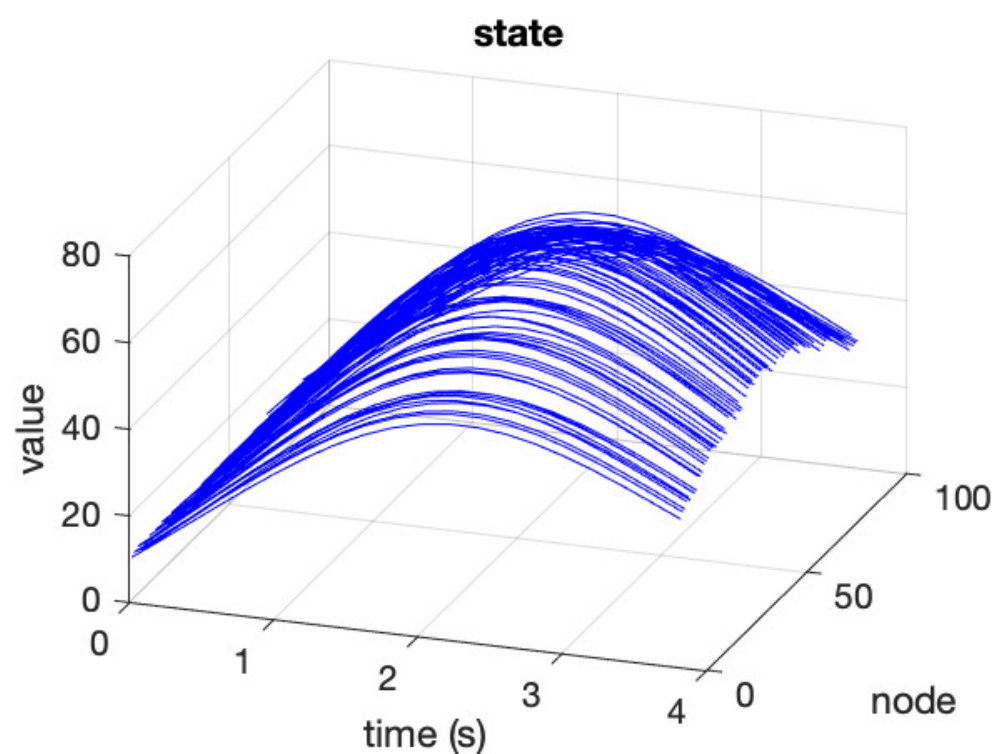}}~~
    \subfloat[GFG best response]{\includegraphics[width=4cm,trim = {0.1cm 0 0.1cm 0}, clip]{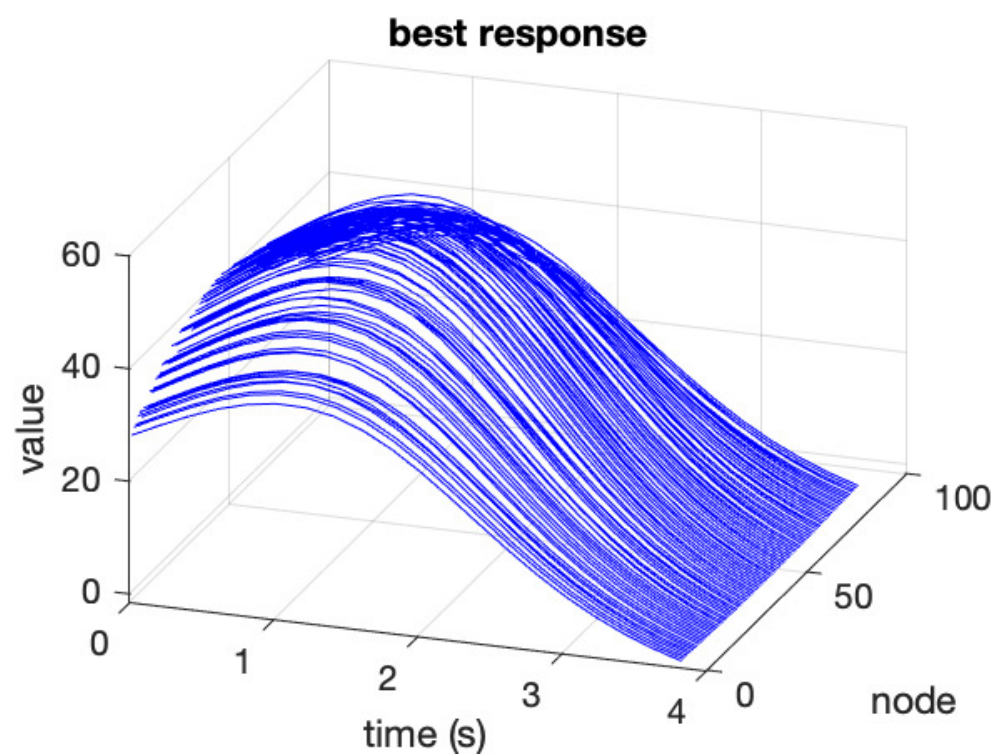}}~~
    \subfloat[Offset process $s$]{\includegraphics[width=4cm,trim = {0.1cm 0 0.1cm 0}, clip]{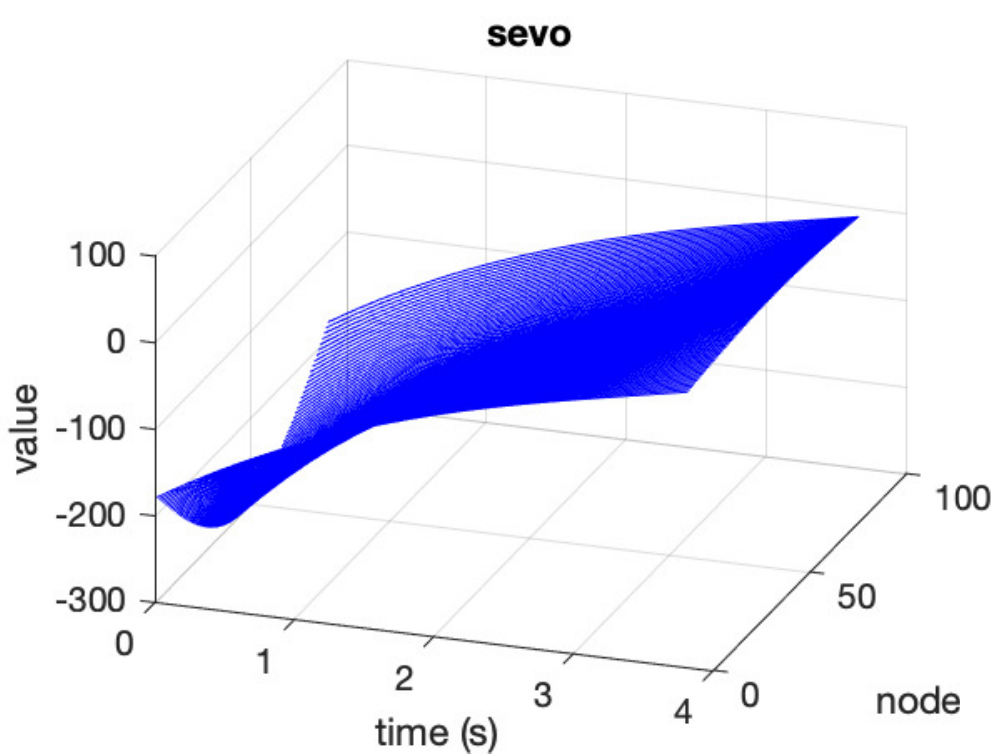}}\\
    \subfloat[Graphon field $z$]{\includegraphics[width=4cm,trim = {0.1cm 0 0.1cm 0}, clip]{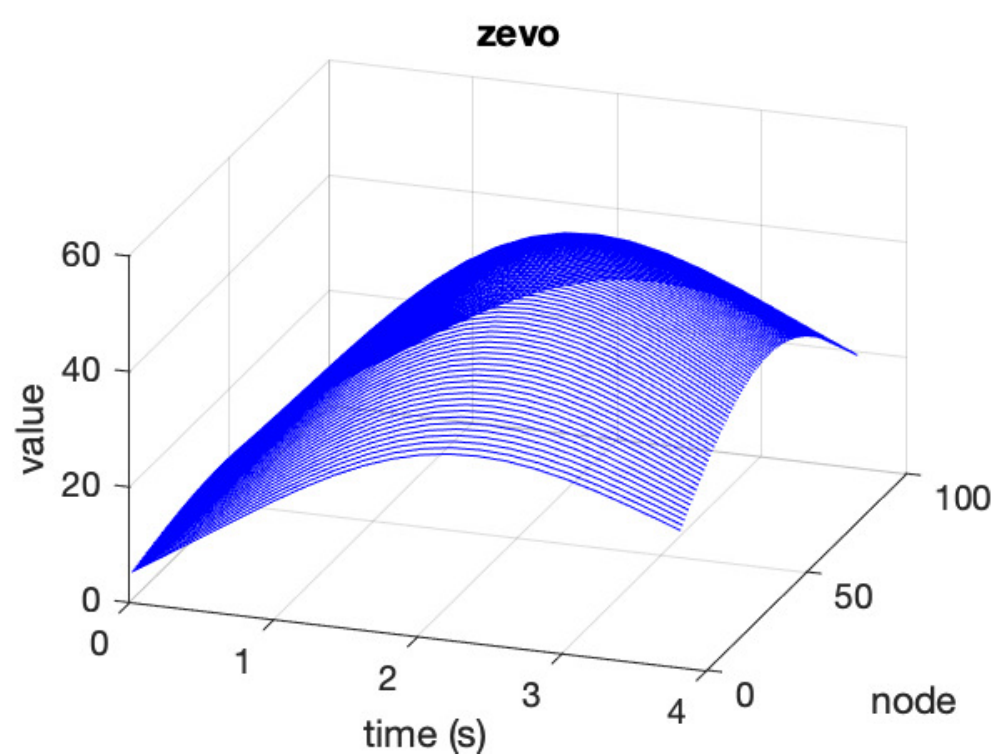}}~~
    \subfloat[Empirical graphon field]{\includegraphics[width=4cm,trim = {0.1cm 0 0.1cm 0}, clip]{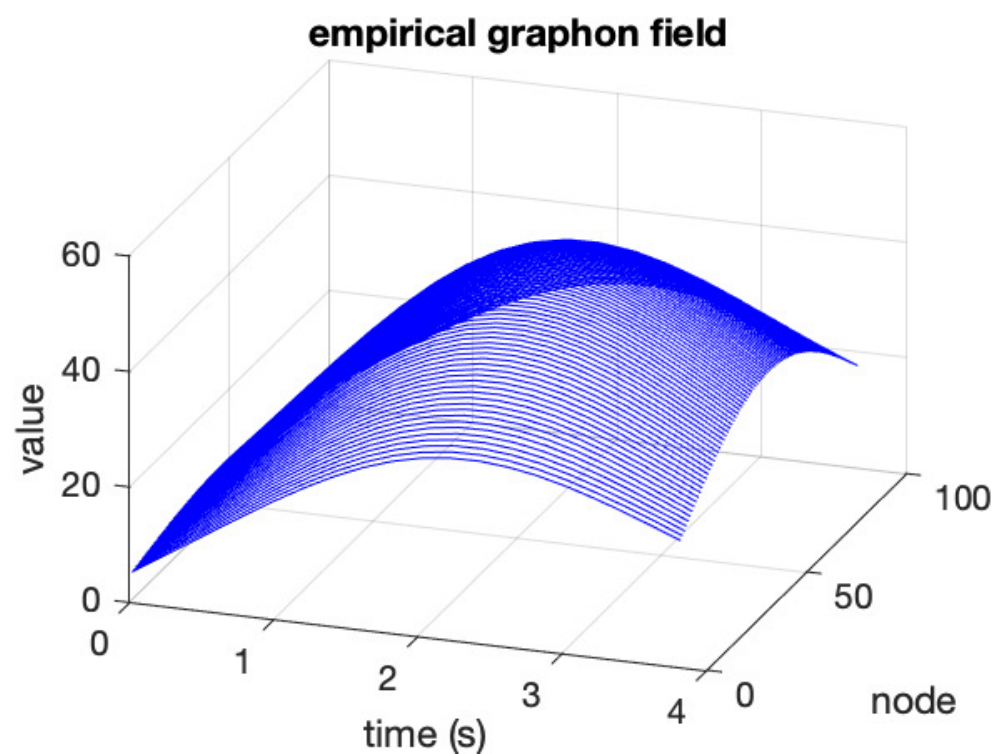}}~~
    \subfloat[Graphon field estimate error]{\includegraphics[width=4cm,trim = {0.1cm 0 0.1cm 0}, clip]{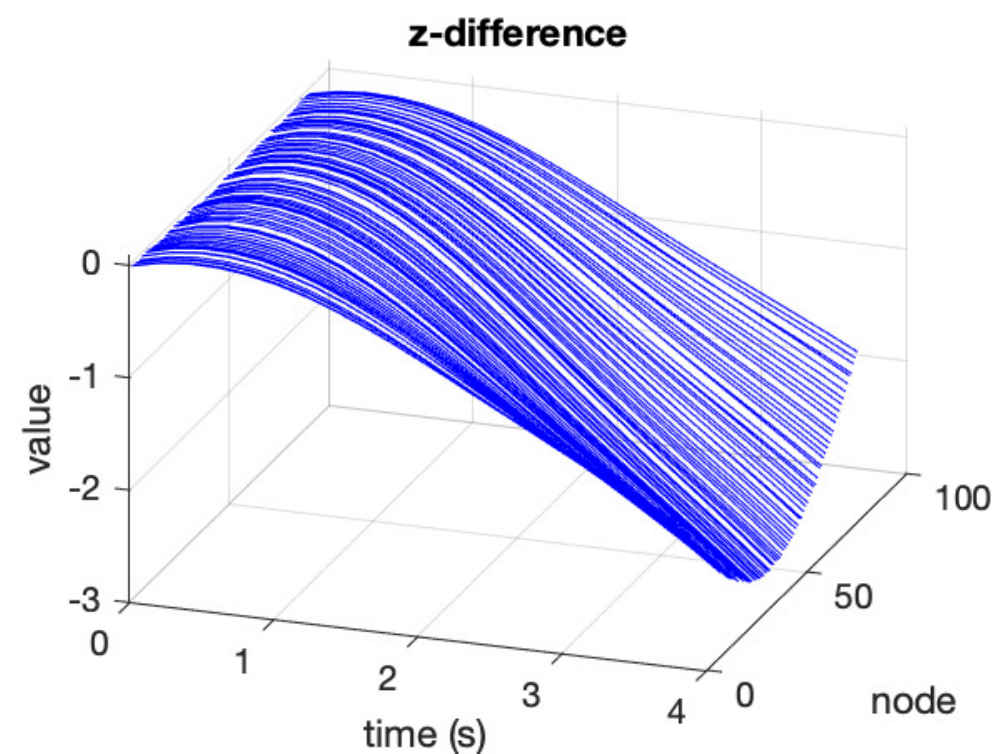}}\\
     \subfloat[Graph structure and its pixel representation]{\includegraphics[width=8.5cm,trim = {1cm 0 1cm 0}, clip]{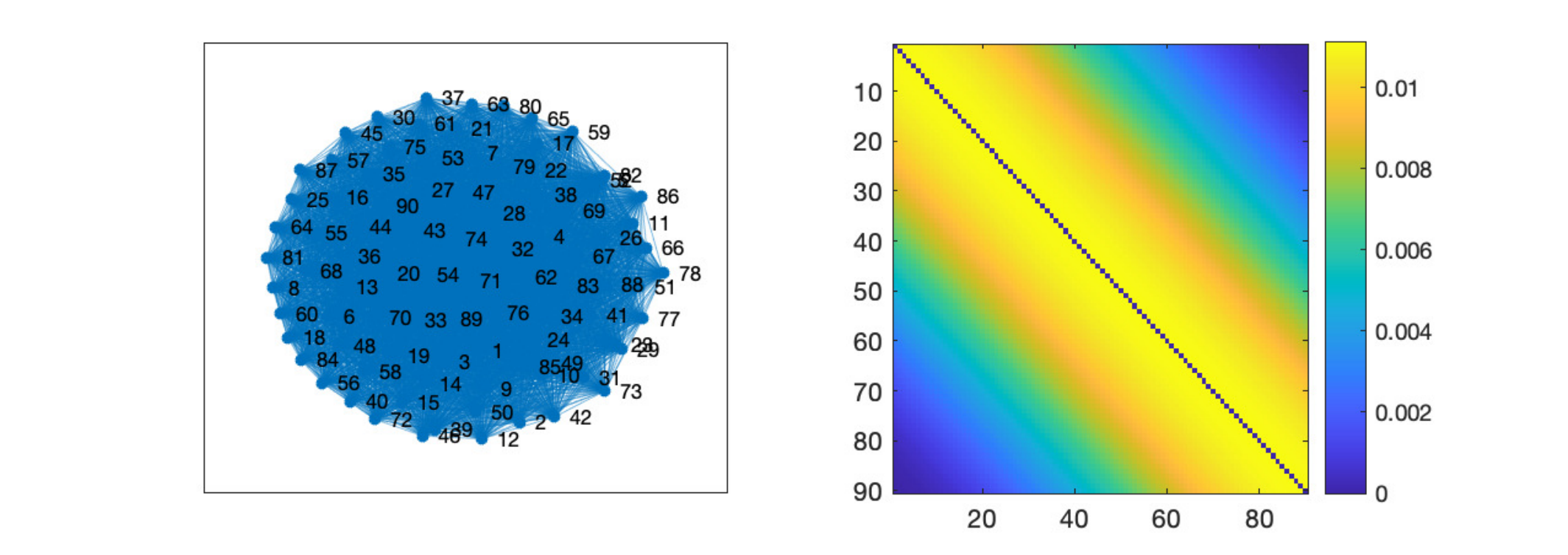}} ~~
    \subfloat[Individual cost]{\includegraphics[width=4cm,trim = {0cm 0 0cm 0}, clip]{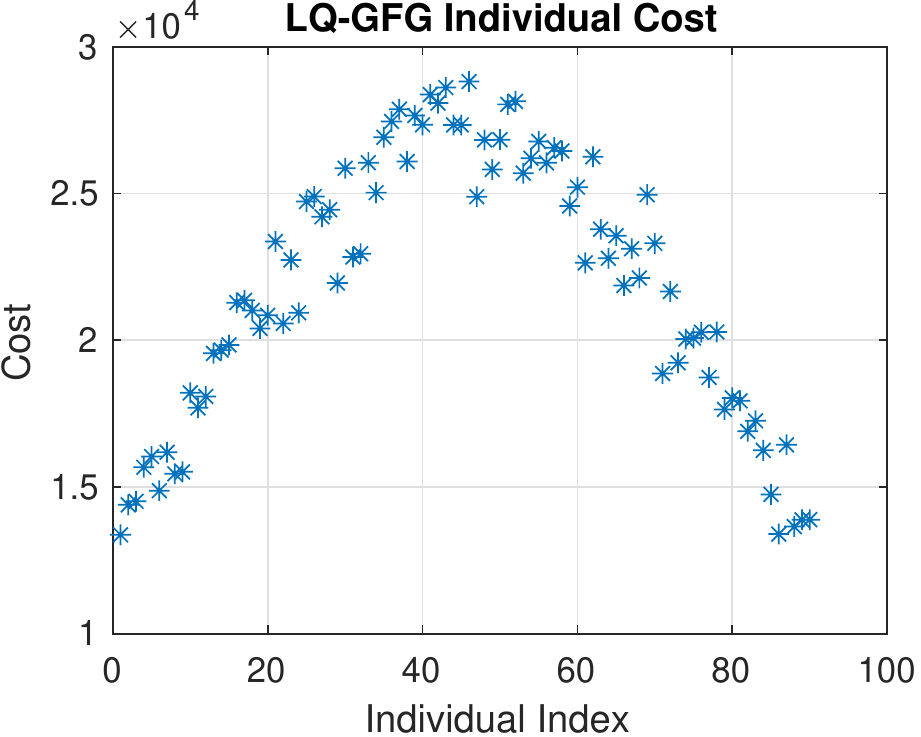}}
         \caption{Graphon field game simulation for systems on weighted graphs of size 90 with weights generated from the graphon $\FA(x,y)= 0.5\cos(\pi (x-y))+0.5$ for all $x, y \in [0,1]$ based on the $90$-uniform partition of $[0,1]$.  }\label{fig:fig-S3Equal}
\end{figure}

\section{Conclusions}
 Future work will be focused on the following  LQ-GFG problem aspects: 
(1) LQ-GFG  problems where the agent dynamics include stochastic disturbances;
(2) LQ-GFG problems based on sampling procedures;
(3) less restrictive conditions for the LQ-GMG $\varepsilon$-Nash property;
(4) an LQ-GFG methodology  for sparse networks; 
(5) convergence properties of finite equilibria to the limit equilibrium; 
{(6) general approximate solution methods to relax the finite-rank  restrictions on graphon limits.}

\bibliographystyle{plain}
\bibliography{mybib}
\address{Department of Electrical and Computer Engineering, McGill University, Montreal, QC, Canada, H3A 0E9\\
\email{sgao@cim.mcgill.ca}}

\address{Department of Electrical and Computer Engineering, McGill University, Montreal, QC, Canada, H3A 0E9\\
\email{fogurine@cim.mcgill.ca}}

\address{Department of Electrical and Computer Engineering, McGill University, Montreal, QC, Canada, H3A 0E9\\
\email{peterc@cim.mcgill.ca}}

\section{Appendix}
\subsection{Proof of Proposition \ref{prop:one}} \label{sec:PropositionProof}
The proof follows immediately from the standard fixed point method and we include it here for convenience of reference.
\begin{proof}

Consider the following transformation: 
\begin{equation}
  \begin{split}
    \tilde{z}_t^\ell &= z_t^\ell \text{exp}\left[{-\int_0^t\Big(\alpha-\frac{\beta^2}{r}\pi_\tau+\eta \lambda_\ell\Big)d\tau}\right]\\
    \tilde{s}_t^\ell &= s_t^\ell \text{exp}\left[{-\int_t^T\Big(\alpha-\frac{\beta^2}{r}\pi_\tau\Big)d\tau}\right]. 
  \end{split}
\end{equation}
From \eqref{eq:fixed-point-eigen}, we obtain
\begin{equation}\label{eq:trans-fixed-point}
  \begin{aligned}
    \dot{\tilde{z}}_t^\ell &= \frac{\beta^2\lambda_\ell}{r} \frac{ 1}{B_\ell(t)} \tilde{s}_t^\ell, \quad \tilde{z}_0^\ell = z_0^\ell \\
    \dot{\tilde{s}}_t^\ell 
                 &= -(1-\eta \pi_t) B_\ell(t) \tilde{z}_t^\ell , \quad \tilde{s}_T^\ell = s_T^\ell   =0  
  \end{aligned}
\end{equation}
where $B_\ell(t)\triangleq\text{exp}\left[{\int_0^t\Big(\alpha-\frac{\beta^2}{r}\pi_\tau+\eta \lambda_\ell\Big)d\tau-\int_t^T\Big(\alpha-\frac{\beta^2}{r}\pi_\tau\Big)d\tau}\right].$
 Note that $B_\ell(t)>0$ holds for all $t \in [0,T]$. 
In a compact form, we have
\begin{equation}
\MATRIX{\dot{\tilde{z}}_t^\ell\\
  \dot{\tilde{s}}_t^\ell } =
\MATRIX{ 0 & \frac{\beta^2}{r} \frac{\lambda_\ell}{B_\ell(t)} \\
-  (1-\eta \pi_t)B_\ell(t) & 0} \MATRIX{\tilde{z}_t^\ell\\
  \tilde{s}_t^\ell },
  ~ 
  \MATRIX{\tilde{z}_0^\ell \\ \tilde{s}_T^\ell 
  } = \MATRIX{z_0^\ell\\s_T^\ell}.
\end{equation}
Let $\BT_g$ and $\BT_b$ be linear mappings from $C([0,T];\BR)$ to $C([0,T];\BR)$ defined as:
\begin{equation}
\begin{aligned}
  \BT_g(x)(t) &= z_0^\ell + \int_0^t \frac{\beta^2}{r}\frac{\lambda_\ell}{B_\ell(\tau)} x(\tau) d\tau,\\
  \BT_b(x)(t) &= 0 + \int_t^T (1-\eta \pi_\tau) B_\ell(\tau)x(\tau) d\tau ,\quad t\in [0,T]
\end{aligned}
\end{equation}
for all $x\in C([0,T];\BR)$. 

Consider the mapping $\BM: C([0,T];\BR) \rightarrow C([0,T];\BR)$  defined by the composition of $\BT_g$ and $\BT_b$ as follows: 
$
 \BM: x \mapsto \BT_g(\BT_b(x)), ~ x \in C([0,T];\BR).
$
We consider the Banach space of continuous functions $C([0,T];\BR)$ endowed with the sup norm 
$\|x\|_\infty = \sup_{t\in [0,T]}|x(t)|$ for any $ x\in C([0,T];\BR)$.
To show the existence of a unique fixed point to \eqref{eq:trans-fixed-point}, we shall identify a condition for which the following holds:
\begin{equation}
\|\BM(x)- \BM(y)\|_{\infty} \leq L\|x-y\|_{\infty},   \quad  L < 1.
\end{equation}
For $t, \tau \in [0,T]$,
\begin{equation}
\begin{aligned}
    |\BM(x)(t)- \BM(y)(t)| &=  \left| \int_0^t \frac{\beta^2}{r} \frac{ \lambda_\ell}{B_\ell(\tau)} \big(\BT_b(x)(\tau)-\BT_b(y)(\tau)\big)d\tau \right| 
    \\
    & \leq  \int_0^t \frac{\beta^2}{r}\frac{|\lambda_\ell|}{B_\ell(\tau)}  \big(\big|\BT_b(x)(\tau)-\BT_b(y)(\tau)\big|\big) d\tau , \\
\end{aligned}
\end{equation}
\begin{equation}
  \begin{aligned}
      \Big|\BT_b(x)(\tau)- \BT_b(y)(\tau)\Big| &= \left| \int_\tau^T (1-\eta \pi_s)B_\ell(s)(x(s)-y(s))ds\right|\\
                 & \leq \int_\tau^T \left|(1-\eta \pi_s)\right|B_\ell(s)ds \|x-y\|_{\infty}.
  \end{aligned}
\end{equation}
Therefore, 
\begin{equation}
\begin{aligned}
  &\|\BM(x)- \BM(y)\|_\infty  \leq \int_0^T \frac{\beta^2}{r}\frac{ |\lambda_\ell|}{B_\ell(\tau)} \left(\int_\tau^T |(1-\eta \pi_s)|B_\ell(s)ds\right) d\tau \|x-y\|_{\infty}.
\end{aligned}
\end{equation}
Hence we obtain the following sufficient condition for a unique fixed point: 
\begin{equation}\label{eq:contraction}
  \int_0^T \frac{\beta^2}{r}\frac{ |\lambda_\ell|}{B_\ell(\tau)} \int_\tau^T |(1-\eta \pi_s)|B_\ell(s)ds  d\tau <1.
\end{equation}

For \eqref{eq:fixed-point}, the fixed point condition needs to be satisfied for all eigendirections. Therefore we obtain the result.
\end{proof}

\subsection{Proof of Theorem \ref{thm:deterministic-LQ-GFG}} \label{sec: proof-deterministic-LQ-GFG}
We define the following functions: for $t, \tau \in [0,T]$,
\begin{align}
&A_c^N(t)\triangleq \Big[(\alpha - \frac{\beta^2}{r} \pi_t) I +\frac{\eta}{N} A_N\Big] \label{eq:AC}\\
  & \Phi(t,\tau)\triangleq\textup{exp}\left({\int_\tau^t A_c^N(s) ds}\right) \label{eq:Phi}\\
  & \Gamma(t,\tau) \triangleq \textup{exp}\left({\int_\tau^t(\alpha-\frac{\beta^2}{r} \pi_s)ds}\right). \label{eq:Gamma}
\end{align}
Clearly $\Phi(t,\tau) = \Gamma(t,\tau) \text{exp}\left(\frac{\eta(t-\tau)}{N} A_N\right),$
where the first part $\Gamma(t,\tau)$ is just a scalar and the second part $\text{exp}\left(\frac{\eta(t-\tau)}{N} A_N\right)$ is a $N\times N$ matrix.
The same definitions are also used in Appendix section \ref{sec:proof-rinit}. 

\begin{proof}
  Let 
$\Fz^\gamma_t\triangleq \sum_{\ell=1}^d z^\ell_t \Ff_\ell(\gamma),~  \gamma \in [0,1], $
be the local graphon field calculated based on \eqref{eq:eigen-mean-state} and $\Fz^\gamma$ denote its trajectory over time $[0,T]$.  Denote the $N$-uniform partition of $[0,1]$ by $\{P_1,...,P_N\}$.
For $P_i \subset [0,1]$, let
$$
\begin{aligned}
   \bar \Fz^{i}_t &\triangleq 
   \frac{1}{\mu(P_i)}\int_{P_i} \Fz^\gamma_t d\gamma  \in \BR
    \quad \text{and} \quad \bar \Fz_t^\gamma \triangleq \bar \Fz_t^i, ~\gamma \in P_i
\end{aligned}
$$
Their trajectories over time $[0,T]$ are respectively denoted by $\bar\Fz^i$ and $\bar \Fz$.  Similarly we define $\bar \Fs_t$, $\bar\Fs_t^i$, $\bar \Fs$ and $\bar\Fs^i$.


%
For agent $i$ corresponding to $P_i \subset [0,1]$, 
the cost induced by following the infinite population and graphon limit Nash law in Strategy \ref{strategy-deterministic}  is given by 
\begin{equation*}
  \begin{aligned}
    &J({u}^{o i},u^{-o i})=  \frac12 \int_0^T \Big((x^{oi}_t -\frac1N\sum_{j=1}^{N}a_{ij}x^{oj}_t)^2 + r(u_t^{oi})^2 \Big)dt\\
    &  = J({u}^{o i},\bar\Fz^i) + \frac12 \int_0^T \Big[(\bar\Fz_t^i-\frac1N\sum_{j=1}^{N}a_{ij}x^{oj}_t)\Big]^2_{({u}^{o i},u^{-o i})}dt\\
    &\qquad  +  \int_0^T \Big[(x^{oi}_t - \bar\Fz_t^i)(\bar\Fz_t^i-\frac1N\sum_{j=1}^{N}a_{ij}x^{oj}_t) \Big]_{({u}^{o i},u^{-o i})}dt \\
    &\triangleq J({u}^{o i},\bar \Fz^i) + \frac12 I_2 + \frac12 I_3,
  \end{aligned}
\end{equation*}
where $x_t^{oi}$ and $u_t^{oi}$ represent respectively the state and the action for agent $i$ at time $t$ under the prescribed $\varepsilon$-Nash law. We use $(u^{oi}, u^{-oi})$ in the subscript to indicate the underlying strategy.
By the Cauchy-Schwarz inequality, we obtain
\begin{equation}\label{eq:I2-I3-relation}
\begin{aligned}
  |I_3| 
  & \leq 2 \sqrt{I_2}\left[\int_0^T (x^{oi}_t - \bar \Fz_t^i)^2 dt\right]^\frac12_{({u}^{oi},u^{-o i})}.
\end{aligned}
\end{equation}
%
Similarly, we obtain
\begin{equation*}
  \begin{aligned}
    &J({u}^{i},u^{-oi}) = \frac12 \int_0^T \Big((x^i_t - \frac1N\sum_{j=1}^{N}a_{ij}x^j_t)^2 + r(u_t^{i})^2 \Big)_{({u}^{i},u^{-oi})}dt\\
        &  = \frac12 \int_0^T \Big[(x^i_t -\bar \Fz_t^i)^2+ r(u_t^{i})^2 \Big]dt+\frac12 \int_0^T \Big[(\bar\Fz_t^i-\frac1N\sum_{j=1}^{N}a_{ij}x^j_t)\Big]^2_{({u}^{i},u^{-o i})}dt\\
    &\qquad  +  \int_0^T \Big[(x^i_t - \bar \Fz_t^i)(\bar \Fz_t^i-\frac1N\sum_{j=1}^{N}a_{ij}x^j_t) \Big]_{({u}^{ i},u^{-o i})}dt \\
    & \triangleq J(u^{i}, \bar\Fz^i) +\frac12 I^\prime_2 + \frac12 I^\prime_3.
  \end{aligned}
\end{equation*}
and 
\begin{equation}
\begin{aligned}
  |I^\prime_3| &\leq 2 \sqrt{I^\prime_2}\left[\int_0^T (x^i_t -\bar \Fz_t^i)^2 dt\right]^\frac12_{({u}^{i},u^{-o i})}.
\end{aligned}
\end{equation}
Therefore, we obtain the following: 
\begin{equation}\label{eq:final-cost-diff}
  \begin{aligned}
    & J(u^{oi}, u^{-oi}) - J(u^{i}, u^{-oi})\\
    &= J(u^{oi}, u^{-oi}) -J(u^{oi}, \bar\Fz^i) + J(u^{oi},\bar \Fz^i) - J(u^{i},\bar \Fz^i)  + J(u^{i}, \bar\Fz^i) - J(u^{i}, u^{-oi})\\
    &{\leq} \frac12 (|I_2| + |I_3|) + |\left(J(u^{oi}, \bar \Fz^i) - J(u^{i}, \bar\Fz^i) \right)|+ \frac12 (|I_2^\prime| + |I_3^\prime|).
  \end{aligned}
\end{equation}
%
{We will establish the following asymptotic estimates for $I_2$, $I_3$, $I_2^\prime$ and $I_3^\prime$: 
for fixed $T> 0$,
\begin{equation}
\begin{aligned}
	  &|I_2| = O \Big(E_N^2 \Big),  \quad |I_3| = \max\left\{O\Big(E_N \Big), O\Big(E_N^2 \Big)\right\} , \\
  &|I_2^\prime|  = O \Big(E_N^2 \Big),\quad |I_3^\prime|  = \max\left\{O\Big(E_N \Big), O\Big(E_N^2 \Big)\right\}, 
\end{aligned}
\end{equation}
where  $E_N\triangleq \max_{1\leq i\leq N} \frac{1}{\mu(P_i)} \left\|  (\FA-\FA^\FN)\mathds{1}_{P_i}\right\|_2.$
These estimates, together with 
Lemma \ref{lem:opt-mfg-cost-comp} and \eqref{eq:final-cost-diff},  imply that
\begin{equation}
\begin{aligned}
  &J(u^{oi}, u^{-oi}) - \inf_{u^i \in \mathcal{U}} J(u^{i}, u^{-oi})  = \max  \left\{O(E_N),O(E_N^2)\right\}.
\end{aligned}
\end{equation}

The rest of the proof will be devoted to establishing the asymptotic estimates for  $I_2$, $I_3$, $I_2^\prime$ and $I_3^\prime$.}

\subsubsection*{{1) Estimate for $I_2$}}
The next step is to show the upper bound for the term $I_2$:
\begin{equation}
  \begin{aligned}
    I_2 & \triangleq \int_0^T \Big[(\bar\Fz_t^i-\frac1N\sum_{j=1}^{N}a_{ij}x^{oj}_t)\Big]^2_{({u}^{o i},u^{-o i})}dt.
  \end{aligned}
\end{equation}
Based on \eqref{eq:fixed-point}, 
we obtain 
\begin{align}
  \dot{\bar \Fz}_t^i&= \Big(\alpha-\frac{\beta^2}{r}\pi_t\Big)\bar\Fz_t^i + \eta \frac{1}{\mu(P_i)}\int_{P_i} [\FA \Fz_t]^\gamma d\gamma  \\
  &~~+ \frac{\beta^2}{r} \frac{1}{\mu(P_i)}  \int_{P_i}[\FA \Fs_t]^\gamma d\gamma, ~ 
     \bar \Fz_0^i = \frac{1}{\mu(P_i)}\int_{P_i}[\FA \Fx_0^\FN]^\gamma d\gamma,\nonumber\\
   \dot{\Fs}_t^i&=-\Big(\alpha-\frac{\beta^2}{r}\pi_t\Big)\bar \Fs_t^i - (1-\eta \pi_t)\bar \Fz_t^i,  \quad  \bar\Fs^i_T = 0.
  \end{align}

The closed loop dynamics for the $i$th agent under Strategy \ref{strategy-deterministic} in  the finite population problem is given by  
  \begin{align}
  \dot{x}_t^{oi} &= \Big(\alpha - \frac{\beta^2}{r} \pi_t\Big){x}^{oi}_t +\eta \frac1N\sum_{j=1}^{N}a_{ij}x^{oj}_t+ \frac{\beta^2}{r} \bar \Fs_t^i.
  \end{align}
Let ${\Fz}_t^{o\FN} \in L^2[0,1]$ denote the step function that corresponds to the vector $[z_t^{o1} \ldots z_t^{oN}]^\TRANS$.  
%
%
Let $\Delta_t^{Ni} \triangleq \bar \Fz_t^i - z_t^{oi}.$ Then
$$
\begin{aligned}
  \dot \Delta_t^{Ni} 
  =& \Big(\alpha - \frac{\beta^2}{r} \pi_t\Big)\Delta_t^{Ni} 
  +\frac{\beta^2}{r}  \frac{1}{\mu(P_i)}\int_{P_i} \Big[\FA\Fs_t-\FA^\FN\bar\Fs_t\Big]^\gamma d\gamma,\\
  & + \eta\frac{1}{\mu(P_i)}\int_{P_i}\big[\FA  \Fz_t - \FA^\FN \Fz_t^{o\FN}\big]^\gamma d\gamma  \\
  %
  =& \Big(\alpha - \frac{\beta^2}{r} \pi_t\Big)\Delta_t^{Ni} 
  +\frac{\beta^2}{r} \frac{1}{\mu(P_i)} \Big\langle  \mathds{1}_{P_i}, (\FA-\FA^\FN)\Fs_t\Big\rangle\\
  &+ \eta\frac{1}{\mu(P_i)}\left\langle \mathds{1}_{P_i},(\FA - \FA^\FN )\Fz_t\right\rangle %
  + \eta\frac1N \sum_{j=1}^N a_{ij} \Delta_t^{Nj}.
\end{aligned}
$$
Therefore, 
\begin{equation}\label{eq:Delta}
  \begin{aligned}
    \dot{\Delta}_t^N =  \Big[(\alpha - \frac{\beta^2}{r} \pi_t) I +\frac{\eta}{N} A_N\Big] {\Delta}_t^N + \frac{\beta^2}{r}D_t^{N\Fs} + \eta D_t^{N\Fz}
  \end{aligned}
\end{equation}
where 
$$
\begin{aligned}
D_t^{N\Fs} &= \left(\frac{1}{\mu(P_i)} \Big\langle  \mathds{1}_{P_i}, (\FA-\FA^\FN)\Fs_t\Big\rangle\right)_{i=1}^N \in \BR^N,\\
D_t^{N\Fz} &= \left(\frac{1}{\mu(P_i)} \Big\langle  \mathds{1}_{P_i},(\FA-\FA^\FN) \Fz_t\Big\rangle\right)_{i=1}^N\in \BR^N.
\end{aligned}
$$
The initial condition is given by $\Delta_0^N =[\Delta_0^{N1},\ldots, \Delta_0^{NN}]^\TRANS$ where
\begin{equation}\label{eq:Delta-initial-cond}
\begin{aligned}
  \Delta_0^{Ni} &= \bar \Fz_0^i - z_0^{oi}= \frac{1}{\mu{(P_i)}} \big\langle \mathds{1}_{P_i}, (\FA- \FA^\FN) \Fx_0^\FN \big \rangle.
\end{aligned}
\end{equation}
We want to ensure that $\|\Delta_t^{N}\|_{\infty}\triangleq \max_{i}\{|\Delta_t^{Ni}|\}$ is bounded and establish the rate of convergence with respect to $N$.
The solution  $\{\Delta_t^N, t \in [0,T]\}$ to \eqref{eq:Delta} with the initial condition specified by \eqref{eq:Delta-initial-cond} is given as follows:
\begin{equation}
  \Delta_t^N  = \Phi(t,0)\Delta_0^N + \int_0^t\Phi(t,\tau) \Big(\frac{\beta^2}{r}D_\tau^{N\Fs}+ \eta D_t^{N\Fz} \Big)d\tau .
\end{equation}
%
%
By the definition of  $\|\cdot\|_\infty$,  we obtain
\begin{equation} \label{eq:Delta-infty}
\begin{aligned}
  \|\Delta_t^N\|_\infty& 
 \leq
  \left\|\Phi(t,0)\right\|_{\infty}\|\Delta_0^N\|_\infty + \int_0^t\left\| \Phi(t,\tau) \Big(\frac{\beta^2}{r}D_\tau^{N\Fs}+ \eta D_\tau^{N\Fz} \Big)\right\|_\infty d\tau  \\
  &\leq 
  \Big\|\Phi(t,0)\Big\|_{\infty} E_N \|\Fx_0^\FN\|_2 + \int_0^t \left\|\Phi(t,\tau)\right\|_\infty E_N \Big(\frac{\beta^2}{r}\|\Fs_\tau\|_2+|\eta|\|\Fz_\tau\|_2\Big)d\tau %
\end{aligned}
\end{equation}
where $E_N\triangleq \max_{1\leq i\leq N} \frac{1}{\mu(P_i)} \left\|  (\FA-\FA^\FN)\mathds{1}_{P_i}\right\|_2.$
Furthermore, $ \left\|\Phi(t,\tau)\right\|_\infty$ with $t\geq \tau$ is uniformly bounded in $N$ since
\begin{equation}
  \begin{aligned}
    \left\|\Phi(t,\tau)\right\|_\infty
      & = \Gamma(t,\tau)\left\|e^{(t-\tau) \frac{\eta}{N} A}\right\|_\infty 
       \leq \Gamma(t,\tau) e^{(t-\tau) |\eta| c}\\
  \end{aligned}
\end{equation}
where $|a_{ij}|\leq c$ as in Assumption \ref{ass:a_ii=0}.
Since $\left\|\Phi(t,\tau)\right\|_\infty$, $\|\Fx_0^\FN\|_2$, $\|\Fs_\tau\|_2$ and $\|\Fz_\tau\|_2$ are all uniformly bounded in $N$, \eqref{eq:Delta-infty} implies  
\begin{equation}\label{eq:Delta-BigO}
   \|\Delta_t^N\|_\infty   = O\Big(E_N\Big).
\end{equation}
Hence, by the definition of $I_2$, for fixed $T> 0$,
\begin{equation}\label{eq:I2-bigO}
  |I_2| = O \Big(E_N^2 \Big).
\end{equation}
\subsubsection*{2) Estimate for $I_3$}
Similarly, letting $\delta_t^{Ni} \triangleq \bar \Fz^i_t- x_t^{oi}$ we obtain 
\begin{equation}
  \begin{aligned}
      \dot \delta_t^{Ni} =& (\alpha - \frac{\beta^2}{r} \pi_t)\delta_t^{Ni} 
  +\frac{\beta^2}{r} \left[ \frac{1}{\mu(P_i)}\int_{P_i}\big[\FA \Fs_t\big]^\gamma d\gamma- \bar \Fs_t^i\right]\\
  & + \eta \left[\frac{1}{\mu(P_i)}\int_{P_i} [ \FA \Fz_t]^\gamma d\gamma -\frac1N\sum_{j=1}^{N}a_{ij}x^{oj}_t\right]\\
 =& (\alpha - \frac{\beta^2}{r} \pi_t)\delta_t^{Ni} 
   +\frac{\beta^2}{r} \frac{1}{\mu(P_i)}
\langle(\FA-\BI)\mathds{1}_{P_i}, \Fs_t \rangle\\
& + \eta \frac{1}{\mu(P_i)} \langle  (\FA- \BI)\mathds{1}_{P_i},\Fz_t\rangle + \eta \Delta_t^{Ni}.
  \end{aligned}
\end{equation}
Therefore, 
\begin{equation*}
  \begin{aligned}
       \delta_t^{Ni} &= 
       \Gamma(t,0)\delta_0^{Ni} + \int_0^t \Gamma(t,\tau) \frac{\beta^2}{r} \frac{1}{\mu(P_i)} \langle(\FA-\BI)\mathds{1}_{P_i}, \Fs_\tau\rangle d\tau \\
       &\quad+ \int_0^t \Gamma(t,\tau) \frac{\eta }{\mu(P_i)} \langle (\FA-\BI) \mathds{1}_{P_i},\Fz_\tau\rangle  d\tau + \int_0^t \Gamma(t,\tau) \eta \Delta_\tau^{Ni}  d\tau \\
         &\leq \Gamma(t,0)\|\FA-\BI\|_\text{op}\|\Fx_0^\FN\|_2 \\
         &\quad+ \int_0^t \Gamma(t,\tau) \|(\FA-\BI)\|_\text{op} \Big(\frac{\beta^2}{r}\|\Fs_\tau\|_2 +\eta \|\Fz_\tau\|_2 \Big)  d\tau + \int_0^t \Gamma(t,\tau) \eta \Delta_\tau^{Ni}  d\tau \\
    \end{aligned}
\end{equation*}
with the initial condition
$$
\begin{aligned}
\delta_0^{Ni} & = \bar \Fz^i_0- x_0^{oi} 
 =  \frac{1}{\mu(P_i)} \langle (\FA - \BI)\mathds{1}_{P_i} , \Fx_0^\FN  \rangle.  \\
\end{aligned}
$$
Furthermore, since $\Fx_0^\FN$, $\Fz_t$, $\Fs_t$, $\FA$, $\|\FA-\BI\|_{\text{op}}$ are uniformly bounded in $N$ and $\Delta_t^{Ni}$ is of the order $E_N$ for all $i$, we obtain
  $|\delta_t^{Ni}| = \max\left\{O (1), O(E_N)\right\}.$  This together with \eqref{eq:I2-I3-relation} and \eqref{eq:I2-bigO} implies
\begin{equation}\label{eq:I-3}
  |I_3| = \max\left\{O\Big(E_N \Big), O\Big(E_N^2 \Big)\right\} .
\end{equation}

\subsubsection*{3) Estimates for $I_2^\prime$ and $I_3^\prime$}
By Assumption \ref{ass:a_ii=0},  $I_2^\prime = I_2$.
Thus 
\begin{equation}\label{eq:I-2prime}
  |I_2^\prime| =  O(E_N^2) 
\end{equation}
Similar to \eqref{eq:I-3}, we obtain 
\begin{equation}\label{eq:I-3prime}
 | I_3^\prime| = \max  \left\{O(E_N),O(E_N^2) 
  \right\}. 
\end{equation}

\end{proof}                                             
\subsubsection*{Lemma for  Theorem \ref{thm:deterministic-LQ-GFG}}~\\~

Let $x^{oi}, i\in\{1,...,N\}$, denotes the state trajectory of agent $j$ when all agents are following Strategy \ref{strategy-deterministic}.
Assume the following trajectories are given to an arbitrary agent indexed by $i$:
\begin{itemize}
	\item the reference trajectory $\bar \Fz^i_t\triangleq 
   \frac{1}{\mu(P_i)}\int_{P_i} \Fz^\gamma_t d\gamma$ for all $t\in[0,T]$ with $\Fz$ as the graphon filed of the corresponding limit graphon field game problem satisfying \eqref{eq:fixed-point} 
	\item the dynamic offset $z_t^{oi} \triangleq \frac{1}{N}\sum_{j\in \mathcal{N}_i} a_{ij}x_t^{oj}$ for all $t\in [0,T]$ where $\mathcal{N}_i$ denotes the set of neighours for agent $i$ excluding itself.
\end{itemize}
Then consider the following linear quadratic tracking problem for agent $i$:
\begin{equation}\label{eq:complete-info-lqt}
\begin{aligned}
    &\dot x_t^{i} = \alpha x_t^{i} + \beta u_t^{i}+\eta \frac{1}{N}\sum_{j\in \mathcal{N}_i} a_{ij}x_t^{oj}\\
  &  J(u^{i}, \bar \Fz^i ) = \frac12 \int_0^T\left[ (x^{i}_t-\bar \Fz_t^i)^2 + r(u_t^{i})^2 \right]dt .
\end{aligned}
\end{equation}

\begin{lemma}\label{lem:opt-mfg-cost-comp}
  Under 
  Assumptions \ref{ass:limit-assumptions}, \ref{ass:spectral-graphon}, \ref{ass:Riccati-Sol-Existence}
and \ref{ass:a_ii=0}, the following estimate for the difference between costs based on different control laws for the problem in \eqref{eq:complete-info-lqt} holds
  \begin{equation}\label{eq:cost-opt-track-ave}
  J(u^{oi}, \bar \Fz^i) - \inf_{u^i \in \mathcal{U}}J(u^{i}, \bar \Fz^i)  = O \Big(E_N\Big),
\end{equation}
where $u^{oi}$ is generated based on Strategy \ref{strategy-deterministic}, 
$\mathcal{U}\triangleq L^2([0,T];\BR)$ and  $E_N\triangleq \max_{1\leq i\leq N} \frac{1}{\mu(P_i)} \left\|  (\FA-\FA^\FN)\mathds{1}_{P_i}\right\|_2$.

\end{lemma}

\begin{proof}
The optimal control law for the problem \eqref{eq:complete-info-lqt} is given by
  \begin{align}
      u_t^{*i} &= -\frac{\beta}{r} \pi_t x_t^{*i}+ \frac{\beta}{r} s_t^{*i},\\
      -\dot{\pi}_t &= 2\alpha \pi_t -\frac{\beta^2}{r}\pi_t^2 + 1, \qquad  \pi_T = 0, \label{eq:riccati}\\
      -\dot{s}_t^{*i} & = \Big(\alpha - \frac{\beta^2}{r} \pi_t \Big) s_t^{*i} + \bar\Fz_t^i-\eta \pi_t z_t^{oi},\quad
        s_T^{*i} = 0, \label{eq:s-star}
  \end{align}
   The dynamics and cost under the optimal control  are respectively given by
\begin{align}
    &\dot x_t^{*i} = \alpha x_t^{*i} + \beta u_t^{*i}+\eta \frac{1}{N}\sum_{j\in \mathcal{N}_i} a_{ij}x_t^{oj}\\
  &  J(u^{*i}, \bar \Fz^i ) = \frac12 \int_0^T\left[ (x^{*i}_t-\bar \Fz_t^i)^2 + r(u_t^{*i})^2 \right]dt .
\end{align}
On the other hand, the control following  Strategy \ref{strategy-deterministic} is given by 
\begin{equation}
	      {u}^{oi}_t   =-\frac{\beta}{r} \pi_t x_t^{oi}+ \frac{\beta}{r} \bar\Fs_t^i 
\end{equation}
with $\bar \Fs^i$ defined as in \eqref{eq:Nash-Control-Case12} and $\pi_{(\cdot)}$ given by \eqref{eq:riccati}. Assumptions \ref{ass:limit-assumptions}, \ref{ass:spectral-graphon} and \ref{ass:Riccati-Sol-Existence} ensure that $\bar \Fs^i$ always exists. 
The associated dynamics and cost for agent $i$ are then respectively given by 
\begin{align}
   &\dot x_t^{oi} = \alpha x_t^{oi} + \beta u_t^{oi}+\eta \frac{1}{N}\sum_{j\in \mathcal{N}_i} a_{ij}x_t^{oj}\\
	&   J(u^{oi}, \bar \Fz^i ) = \frac12 \int_0^T\left[ (x^{oi}_t-\bar \Fz_t^i)^2 + r(u_t^{oi})^2 \right]dt.
\end{align}
Based on the definition of $\bar\Fs_t^i$ in \eqref{eq:Nash-Control-Case12} and that of $\bar\Fz_t^i$, it is obvious that 
\begin{equation}
\begin{aligned}
    \dot{\bar \Fs}_t^i&=-\Big(\alpha-\frac{\beta^2}{r}\pi_t\Big)\bar \Fs_t^i - (1-\eta \pi_t)\bar \Fz_t^i,  \quad  \bar\Fs^i_T = 0.
\end{aligned}
\end{equation}
This together with \eqref{eq:s-star} yields
\begin{equation}
  \frac{d(\bar \Fs_t^i - s_t^{*i})}{dt} = -\Big(\alpha-\frac{\beta^2}{r}\pi_t\Big) (\bar \Fs_t^i - s_t^{*i}) +\eta \pi_t(\bar \Fz_t^i- z_t^{oi}).
\end{equation}
Let $\Delta_t^{Ni} \triangleq \bar \Fz_t^i - z_t^{oi}.$ Then
$
  (\bar \Fs_t^i - s_t^{*i}) = \int_T^t\Gamma(t,\tau)\eta \pi_\tau \Delta_\tau^{Ni} d\tau
$
and 
\begin{equation}\label{eq:s-diff}
  (\bar \Fs_t^i - s_t^{*i}) = O \Big(E_N\Big), \quad i\in\{1,...,N\}.
\end{equation}
By comparing the following closed-loop dynamics under these two different control laws,
we obtain
\begin{equation*}
\begin{aligned}
    \frac{d(x_t^{oi} - x_t^{*i})}{dt} & = \Big(\alpha -\frac{\beta^2}{r} \pi_t \Big)(x_t^{oi}- x_t^{*i}) + \frac{\beta^2}{r}(\bar \Fs_t^i - s^{*i}_t),~
    x_0^{oi} - x_0^{*i}  = 0.
\end{aligned}
\end{equation*}
Therefore, the  difference is explicitly obtained as
\begin{equation}
\begin{aligned}
  x_t^{oi} - x_t^{*i} 
  & = \frac{\eta\beta^2}{r}\int_0^t \Gamma(t,\tau) \left(\int_T^\tau \Gamma(\tau,q)\eta \pi_q \Delta_q^{Ni} dq\right)d\tau.
\end{aligned}
\end{equation}
Under Assumption \ref{ass:a_ii=0}, $\sum_{j=1}^N a_{ij}x^{oj} = \sum_{j\in \mathcal{N}_i} a_{ij}x^{oj}$ for all $i\in \{1,\ldots,N\}$. Then the result from \eqref{eq:Delta-BigO} applies here, that is, 
$\forall t\in[0,T], ~  |\Delta_t^{Ni}| = O \Big(E_N \Big).$
Hence for any $t\in[0,T]$,
$  |(x_t^{oi} - x_t^{*i})| = O \Big(E_N\Big).$
This, together with \eqref{eq:s-diff}, implies 
$
  |(u_t^{oi} - u_t^{*i})|= O \Big(E_N\Big).$
Hence 
  $\left|J(u^{oi}, \bar \Fz^i) - J(u^{*i}, \bar \Fz^i) \right| = O \Big(E_N \Big),$
 that is, we obtain \eqref{eq:cost-opt-track-ave}.

\end{proof}
%
\section{Appendix II}
\subsection{Proof of Theorem \ref{thm:stochasticinit-LQ-GFG}} \label{sec:proof-rinit} 
\begin{proof}

  Let 
$\Fz^\gamma_t\triangleq \sum_{\ell=1}^d z^\ell_t \Ff_\ell(\gamma),~  \gamma \in [\underline\gamma,\overline\gamma] \subset [0,1], $
be the Local Graphon  Field calculated based on Strategy \ref{str:randinit} and $\Fz^\gamma$ denote its trajectory over time $[0,T]$.  
For $P_i \subset [0,1]$, let
$
   \bar \Fz^{i}_t \triangleq 
   \frac{1}{\mu(P_i)}\int_{P_i} \Fz^\gamma_t d\gamma 
     $ \text{and} $ \bar \Fz_t^\gamma \triangleq \bar \Fz_t^i, ~\gamma \in P_i.
$
Let  $\bar\Fz^i$ and $\bar \Fz$ respectively denote their trajectories over time $[0,T]$.  Similarly we define $\bar \Fs_t$, $\bar\Fs_t^i$, $\bar \Fs$ and $\bar\Fs^i$.

For agent $i$ that corresponds to $P_i \subset [0,1]$, the cost induced by following the limit control law is given by 
\begin{equation*}
  \begin{aligned}
    &J({u}^{o i},u^{-o i})= \frac12  \mathds{E} \int_0^T \Big((x^{oi}_t -\frac1N\sum_{j=1}^{N}a_{ij}x^{oj}_t)^2 + r(u_t^{oi})^2 \Big)dt\\
    &  = J({u}^{o i},\bar\Fz^i) + \frac12 \mathds{E}\int_0^T \Big[(\bar\Fz_t^i-\frac1N\sum_{j=1}^{N}a_{ij}x^{oj}_t)\Big]^2_{({u}^{o i},u^{-o i})}dt\\
    &\qquad  +  \mathds{E}\int_0^T \Big[(x^{oi}_t - \bar\Fz_t^i)(\bar\Fz_t^i-\frac1N\sum_{j=1}^{N}a_{ij}x^{oj}_t) \Big]_{({u}^{o i},u^{-o i})}dt \\
    &\triangleq J({u}^{o i},\bar \Fz^i) + \frac12 I_2 +\frac12  I_3,
  \end{aligned}
\end{equation*}
where $x_t^{oi}$ and $u_t^{oi}$ represent respectively the state and control for agent $i$ at time $t$ under the prescribed $\varepsilon$-Nash law. 
By the Cauchy-Schwarz inequality, we obtain
\begin{equation}
\begin{aligned}
  I_3 
  & \leq 2 \sqrt{I_2}\left[\mathds{E}\int_0^T (x^{oi}_t - \bar \Fz_t^i)^2 dt\right]^\frac12_{({u}^{oi},u^{-o i})}.
\end{aligned}
\end{equation}
%
Let $J({u}^{i},u^{-oi})$ denote the cost for agent $i$ by unilaterally deviating from the prescribed $\varepsilon$-Nash law where all other agents following the prescribed $\varepsilon$-Nash law. Then
\begin{equation*}
  \begin{aligned}
    J({u}^{i},u^{-oi}) &= \frac12 \mathds{E} \int_0^T \Big((x^i_t - \frac1N\sum_{j=1}^{N}a_{ij}x^j_t)^2 + r(u_t^{i})^2 \Big)_{({u}^{i},u^{-oi})}dt\\
        &  =J(u^{i}, \bar\Fz^i)   +\frac12 \mathds{E}\int_0^T \Big[(\bar\Fz_t^i-\frac1N\sum_{j=1}^{N}a_{ij}x^j_t)\Big]^2_{({u}^{i},u^{-o i})}dt\\
    &\qquad  + \mathds{E} \int_0^T \Big[(x^i_t - \bar \Fz_t^i)(\bar \Fz_t^i-\frac1N\sum_{j=1}^{N}a_{ij}x^j_t) \Big]_{({u}^{ i},u^{-o i})}dt \\
    & \triangleq J(u^{i}, \bar\Fz^i) + \frac12 I^\prime_2 + \frac12 I^\prime_3.
  \end{aligned}
\end{equation*}
Similarly, 
\begin{equation}
\begin{aligned}
 | I^\prime_3| &\leq 2 \sqrt{I^\prime_2}\left[\mathds{E}\int_0^T (x^i_t -\bar \Fz_t^i)^2 dt\right]^\frac12_{({u}^{i},u^{-o i})}.\label{eq:I2-I3-relation-rinit}
\end{aligned}
\end{equation}
Therefore, we obtain the following: 
\begin{equation}\label{eq:final-cost-diff-randinit}
  \begin{aligned}
    & J(u^{oi}, u^{-oi}) - J(u^{i}, u^{-oi})\\
    &= J(u^{oi}, u^{-oi}) -J(u^{oi}, \bar\Fz^i) + J(u^{oi},\bar \Fz^i) - J(u^{i},\bar \Fz^i)   + J(u^{i}, \bar\Fz^i) - J(u^{i}, u^{-oi})\\
    & {\leq} \frac12 (|I_2| + |I_3|) + |\left(J(u^{oi}, \bar \Fz^i) - J(u^{i}, \bar\Fz^i) \right)|+ \frac12 (|I_2^\prime| + |I_3^\prime|).
  \end{aligned}
\end{equation}

{We will establish the following asymptotic estimates for $I_2$, $I_3$, $I_2^\prime$ and $I_3^\prime$: 
for fixed $T> 0$, for any $i\in \{1,...,N\}$,
\begin{equation}
\begin{aligned}
	  & |I_2| = \max\Big\{O\Big(\frac{1}{N}\Big), O(E_N^2)\Big\},  \quad  | I_3 |= \max\left\{O\Big(\frac{1}{\sqrt{N}} \Big), O\Big(E_N \Big), O\Big(E_N^2 \Big)\right\}, \\
  &|I_2^\prime| =  \max\Big\{O\Big(\frac{1}{N}\Big), O(E_N^2)\Big\},\quad |I_3^\prime| = \max\left\{O\Big(\frac{1}{\sqrt{N}} \Big), O\Big(E_N \Big), O\Big(E_N^2 \Big)\right\},
\end{aligned}
\end{equation}
where  $E_N\triangleq \max_{1\leq i\leq N} \frac{1}{\mu(P_i)} \left\|  (\FA-\FA^\FN)\mathds{1}_{P_i}\right\|_2.$
These estimates, 
together with Lemma \ref{lem:opt-mfg-cost-comp-randinit} and \eqref{eq:final-cost-diff-randinit},  imply that
\begin{equation*}
\begin{aligned}
  &J(u^{oi}, u^{-oi}) -  \inf_{u^i \in \mathcal{U}}J(u^{i}, u^{-oi}) = \max\left\{O\Big(\frac{1}{\sqrt{N}} \Big), O\Big(E_N \Big), O\Big(E_N^2 \Big)\right\}. 
\end{aligned}
\end{equation*}

The rest of the proof will be devoted to establishing the asymptotic estimates for  $I_2$, $I_3$, $I_2^\prime$ and $I_3^\prime$.}

\subsubsection*{i) Estimate for $I_2$}
The next step is to show the upper bound for the term $I_2$:
\begin{equation}
  \begin{aligned}
    I_2 & \triangleq \mathds{E} \int_0^T \Big[(\bar\Fz_t^i-\frac1N\sum_{j=1}^{N}a_{ij}x^{oj}_t)\Big]^2_{({u}^{o i},u^{-o i})}dt.
  \end{aligned}
\end{equation}
Similar to the deterministic case, we obtain 
\begin{align}
  \dot{\bar \Fz}_t^i&= \Big(\alpha-\frac{\beta^2}{r}\pi_t\Big)\bar\Fz_t^i + \eta \frac{1}{\mu(P_i)}\int_{P_i} [\FA \Fz_t]^\gamma d\gamma  + \frac{\beta^2}{r} \frac{1}{\mu(P_i)}  \int_{P_i}[\FA \Fs_t]^\gamma d\gamma, ~ \\
  \bar \Fz_0^i & = \frac{1}{\mu(P_i)}\int_{P_i}[\FA \mu \mathds{1}]^\gamma d\gamma,\nonumber\\
   \dot{\bar \Fs}_t^i&=-\Big(\alpha-\frac{\beta^2}{r}\pi_t\Big)\bar \Fs_t^i - (1-\eta \pi_t)\bar \Fz_t^i,  \quad  \bar\Fs^i_T = 0.
  \end{align}

The closed loop dynamics for the finite population problem is given by  
  \begin{align}
  \dot{x}_t^{oi} &= \Big(\alpha - \frac{\beta^2}{r} \pi_t\Big){x}^{oi}_t +\eta \frac1N\sum_{j=1}^{N}a_{ij}x^{oj}_t+ \frac{\beta^2}{r} \bar \Fs_t^i,
  \end{align}
Let ${\Fz}_t^{o\FN} \in L^2[0,1]$ denote the step function that corresponds to the vector $[z_t^{o1} \ldots z_t^{oN}]^\TRANS$ via $N$-uniform partition.  Then
\begin{equation}
    \begin{aligned}
       \dot{\Fz}_t^{o\FN} = \Big(\alpha - \frac{\beta^2}{r} \pi_t\Big){\Fz}^{o\FN}_t +\eta \FA^\FN \Fz^{o\FN}_t+ \frac{\beta^2}{r}\FA^\FN \bar \Fs_t^i,
       \quad \Fz_0^{o\FN}  = \FA^\FN \mu \mathds{1}.
\end{aligned} 
\end{equation}
Let $\Delta_t^{Ni} \triangleq \bar \Fz_t^i - z_t^{oi}.$
$$
\begin{aligned}
  \dot \Delta_t^{Ni} 
  =& \Big(\alpha - \frac{\beta^2}{r} \pi_t\Big)\Delta_t^{Ni} 
  +\frac{\beta^2}{r} \frac{1}{\mu(P_i)} \Big\langle  \mathds{1}_{P_i}, (\FA-\FA^\FN)\Fs_t\Big\rangle,\\
  &+ \eta\frac{1}{\mu(P_i)}\left\langle \mathds{1}_{P_i},(\FA - \FA^\FN )\Fz_t\right\rangle %
  + \eta\frac1N \sum_{j=1}^N a_{ij} \Delta_t^{Nj}.
\end{aligned}
$$
Therefore, 
\begin{equation}
  \begin{aligned}
    \dot{\Delta}_t^N =  \Big[(\alpha - \frac{\beta^2}{r} \pi_t) I +\frac{\eta}{N} A\Big] {\Delta}_t^N + \frac{\beta^2}{r}D_t^{N\Fs} + \eta D_t^{N\Fz}
  \end{aligned}
\end{equation}
where 
$$
\begin{aligned}
D_t^{N\Fs} &= \left(\frac{1}{\mu(P_i)} \Big\langle  \mathds{1}_{P_i}, (\FA-\FA^\FN)\Fs_t\Big\rangle\right)_{i=1}^N \in \BR^N,\\
D_t^{N\Fz} &= \left(\frac{1}{\mu(P_i)} \Big\langle  \mathds{1}_{P_i},(\FA-\FA^\FN) \Fz_t\Big\rangle\right)_{i=1}^N\in \BR^N.
\end{aligned}
$$
By the Cauchy-Schwartz inequality,  
\begin{equation}\label{eq:Dz-Ds-inequalities}
    \begin{aligned}
        D^{Ni}_\Fs (t) &\triangleq \frac{1}{\mu(P_i)} \Big\langle  \mathds{1}_{P_i}, (\FA-\FA^\FN)\Fs_t\Big\rangle \leq E_N \|\Fs_t\|_2,\\
        D^{Ni}_\Fz (t) &\triangleq \frac{1}{\mu(P_i)} \Big\langle  \mathds{1}_{P_i}, (\FA-\FA^\FN)\Fz_t\Big\rangle \leq E_N \|\Fz_t\|_2.
    \end{aligned}
\end{equation}
The initial condition is given by $\Delta_0^N =[\Delta_0^{N1},\ldots, \Delta_0^{NN}]^\TRANS$ where
\begin{equation}\label{eq:Delta-rand}
\begin{aligned}
 \Delta_0^{Ni}  & = \bar \Fz_0^i - z_0^{oi}= \frac{1}{\mu{(P_i)}}\langle \mathds{1}_{P_i}, \FA \mu \mathds{1}- \FA^\FN \Fx_0^\FN \rangle\\
  & = \frac{1}{\mu{(P_i)}}\langle \mathds{1}_{P_i}, \FA(\mathds{1}\mu-\Fx_0^\FN)\rangle + \frac{1}{\mu{(P_i)}}\langle \mathds{1}_{P_i}, (\FA-\FA^\FN) \Fx_0^\FN \rangle.\\
\end{aligned}
\end{equation}
%
The solution  $\{\Delta_t^N, t \in [0,T]\}$ to \eqref{eq:Delta-rand} is given by
\begin{equation}\label{eq:Delta-Evo}
\begin{aligned}
     \Delta_t^N  
  & = \Phi(t,0)\Delta_0^N + \int_0^t \Phi(t,\tau)\Big(\frac{\beta^2}{r}D_\tau^{N\Fs}+ \eta D_t^{N\Fz} \Big)d\tau.
\end{aligned}
\end{equation}
%
Taking the expectation of \eqref{eq:Delta-rand} yields
$$
\begin{aligned}
\mathds{E}\Delta_0^{Ni} 
&= \frac{1}{\mu{(P_i)}}\langle \mathds{1}_{P_i}, (\FA - \FA^\FN) \mu \mathds{1} \rangle=\frac{1}{\mu{(P_i)}}\langle (\FA - \FA^\FN) \mathds{1}_{P_i}, \mu \mathds{1} \rangle,
\end{aligned}
$$
\text{Hence by the Cauchy-Schwarz inequality }
\begin{equation}\label{eq:Delta-Expectation-Ni}
\begin{aligned}
&|\mathds{E}\Delta_0^{Ni}| \leq \frac1{\mu(P_i)} \|(\FA-\FA^\FN)\mathds{1}_{P_i}\|_2|\mu|
 =|\mu| E_N.
\end{aligned}
\end{equation} 
Taking the expectation of \eqref{eq:Delta-Evo} yields
\begin{equation}
    \begin{aligned}
        \mathds{E}\Delta_t^N  
        =& \Phi(t,0)\mathds{E}\Delta_0^N 
        + \int_0^t \Phi(t,\tau) \Big(\frac{\beta^2}{r}D_\tau^{N\Fs}+ \eta D_t^{N\Fz} \Big)d\tau .
    \end{aligned}
\end{equation}
Following a similar argument in the deterministic case  in \eqref{eq:Delta-infty},
\begin{equation} \label{eq:expectation-Delta-randinit}
	\|\mathds{E}\Delta_t^N \|_\infty =\max_i \mathds{E}\Delta_t^{Ni}= O(E_N).
\end{equation}
The variance of $\Delta_0^{Ni}$, $i\in \{1,..., N\}$, satisfies the following
$$
\begin{aligned}
\text{var}(\Delta_0^{Ni}) &= \mathds{E}\left[\frac{1}{\mu{(P_i)}}\langle \mathds{1}_{P_i}, \FA^\FN(\mu \mathds{1}- \Fx_0^\FN) \rangle\right]^2\\
& = \mathds{E} \left[\frac1N \sum_{j=1}^Na_{ij}(x_0^j-\mu)\right]^2  \leq \frac1N \sigma^2 c^2
 \end{aligned}
$$
where $c=\max_{i,j}|a_{ij}|$. Furthermore,
$$
\mathds{E} (\Delta_0^{Ni}-\mathds{E}\Delta_0^{Ni}) (\Delta_0^{Nj}-\mathds{E}\Delta_0^{Nj}) \leq \frac1N \sigma^2 c^2, \quad \forall i, j \in \{1,\ldots,N\}.
$$
Note that
\begin{equation}
    \text{var}(\Delta_0^N)= \frac{A_N}{N}\text{var}(x_0^N) \frac{A_N}{N}^\TRANS= \frac{A_N}{N} \text{diag}(\sigma^2, \ldots, \sigma^2) \frac{A_N}{N}^\TRANS 
\end{equation}
and
\begin{equation} \label{eq:Delta-Init-2ndM}
\mathds{E}\big[\Delta_0^N \Delta_0^N\strut^\TRANS \big] =   \text{var}(\Delta_0^N) + \mathds{E}\Delta_0^N \mathds{E}\Delta_0^N\strut^\TRANS.
\end{equation}
From \eqref{eq:Delta-Evo} and \eqref{eq:Delta-Init-2ndM}, we obtain
\begin{equation}\label{eq:DeltaSqError}
\begin{aligned}
     \mathds{E}&\big[\Delta_t^N \Delta_t^N\strut^\TRANS \big]
    = \Phi(t,0) \left(\text{var}(\Delta_0^N)\right) \Phi(t,0)^\TRANS+ \Phi(t,0) \left(\mathds{E}\Delta_0^N \mathds{E}\Delta_0^N\strut^\TRANS\right) \Phi(t,0)^\TRANS \\
    &+ \mathds{E}\Delta_0^N \cdot \left[\int_0^t \Phi(t,\tau) \Big(\frac{\beta^2}{r}D_\tau^{N\Fs}+ \eta D_t^{N\Fz} \Big)d\tau \right]^\TRANS\\
    & + \left[\int_0^t \Phi(t,\tau) \Big(\frac{\beta^2}{r}D_\tau^{N\Fs}+ \eta D_t^{N\Fz} \Big)d\tau \right]\cdot\mathds{E}\Delta_0^N\strut^\TRANS\\
    & + \left[\int_0^t \Phi(t,\tau) \Big(\frac{\beta^2}{r}D_\tau^{N\Fs}+ \eta D_t^{N\Fz} \Big)d\tau \right] \left[\int_0^t \Phi(t,\tau) \Big(\frac{\beta^2}{r}D_\tau^{N\Fs}+ \eta D_t^{N\Fz} \Big)d\tau \right]^\TRANS\\
    & \triangleq Y_1 + Y_2+ Y_3 + Y_4+ Y_5.
\end{aligned}
\end{equation}
%
The first part of \eqref{eq:DeltaSqError} sastisfies:
\begin{equation}
\begin{aligned}
    \left|\left[Y_1\right]_{ii}\right|  &=\left| \left[\Phi(t,0) \frac{A_N}{N} \text{diag}(\sigma^2, \ldots, \sigma^2) \frac{A_N}{N}^\TRANS  \Phi(t,0)^\TRANS\right]_{ii}\right|\\
    &= \sigma^2\Big[\Gamma(t,0)\Big]^2\left|\left[e^{\left(\frac{\eta t}{N} A_N\right)} \frac{A_N}{N}  {\frac{A_N}{N}}^\TRANS  e^{\left(\frac{\eta t}{N} A_N\right)}\strut^\TRANS\right]_{ii}\right|\\
    & \leq \sigma^2\Big[\Gamma(t,0)\Big]^2\sum_{k=1}^N \Big(\frac{c}{N}e^{c|\eta|t}\Big)^2 \quad  \text{(by Lemma \ref{eq:point-bound-expA})}\\
    & = \frac{\sigma^2}{N}\Big[\Gamma(t,0)\Big]^2 c^2e^{2c|\eta|t}.
\end{aligned}
\end{equation}
The second part of \eqref{eq:DeltaSqError} sastisfies:
$$
\begin{aligned}
    &\left|\left[Y_2\right]_{ii}\right| = \Big[\Gamma(t,0)\Big]^2\left|\left[e^{\left(\frac{\eta t}{N} A_N\right)} \left(\mathds{E}\Delta_0^N \mathds{E}\Delta_0^N\strut^\TRANS\right)  e^{\left(\frac{\eta t}{N} A_N\right)}\strut^\TRANS\right]_{ii}\right|,
\end{aligned}
$$
where
$$
\begin{aligned}
    &\left|\left[e^{\left(\frac{\eta t}{N} A_N\right)} \left(\mathds{E}\Delta_0^N \mathds{E}\Delta_0^N\strut^\TRANS\right)  e^{\left(\frac{\eta t}{N} A_N\right)}\strut^\TRANS\right]_{ii}\right| \\
    &\leq \left|\left[\Big(e^{\frac{\eta t}{N} A_N}-I \Big)\left(\mathds{E}\Delta_0^N \mathds{E}\Delta_0^N\strut^\TRANS\right)  \Big(e^{\frac{\eta t}{N} A_N}-I \Big)^\TRANS\right]_{ii}\right| \\
    &\quad + \left|\left[ \left(\mathds{E}\Delta_0^N \mathds{E}\Delta_0^N\strut^\TRANS\right)  e^{\frac{\eta t}{N} A}_N\strut^\TRANS\right]_{ii}\right| + \left|\left[e^{\frac{\eta t}{N} A_N} \left(\mathds{E}\Delta_0^N \mathds{E}\Delta_0^N\strut^\TRANS\right)\right]_{ii}\right|+ \left|\left[\mathds{E}\Delta_0^N \mathds{E}\Delta_0^N\strut^\TRANS\right]_{ii}\right|\\
    & \text{(by Lemma \ref{eq:lem-expA-Bij} and Lemma \ref{eq:lem-expA-Wij})}\\
    & \leq (e^{|\eta|c t}-1)^2 \mu^2 E_N^2 + 2(e^{|\eta|c t}) \mu^2 E_N^2 + \mu^2 E_N^2= \mu^2 E_N^2(e^{2|\eta|c t}+2) 
\end{aligned}
$$
By Lemma \ref{lem:Ex-vector} and equation \eqref{eq:Dz-Ds-inequalities}, the third  part of \eqref{eq:DeltaSqError} satisfies:
\begin{equation}
    \begin{aligned}
            \left| \left[ Y_3 ^\TRANS\right]_{ii} \right|
            & =\Big|\Big[ \int_0^t  \mathds{E}\Delta_0^N \Gamma(t,\tau)  e^{\left(\frac{\eta(t-\tau)}{N} A_N\right)} \Big(\frac{\beta^2}{r}D_\tau^{N\Fs}+ \eta D_\tau^{N\Fz} \Big)\Big]^\TRANS d\tau \Big]_{ii}\Big|\\
            & 
            \leq |\mu| E_N^2  \int_0^t \Gamma(t,\tau)\left[e^{|\eta|(t-\tau)c}\Big(\frac{\beta^2}{r}\|\Fs_\tau\|_2+ \eta \|\Fz_\tau\|_2 \Big)\right]^\TRANS d\tau.
    \end{aligned}
\end{equation}
This same bound holds for  the fourth  part $Y_4$  of \eqref{eq:DeltaSqError}.
The last  part $Y_5$ of \eqref{eq:DeltaSqError} sastisfies: 
\begin{equation}
    \begin{aligned}
             \left|\Big[Y_5 \Big]_{ii}\right|
    & = \left|\Big[\int_0^t \int_0^t  \Phi(t,\tau) \Big(\frac{\beta^2}{r}D_\tau^{N\Fs}+ \eta D_t^{N\Fz} \Big)d\tau   \Big(\frac{\beta^2}{r}D_\theta^{N\Fs}+ \eta D_\theta^{N\Fz} \Big)^\TRANS  \Phi(t,\theta)^\TRANS d\theta \Big]_{ii}\right|\\
    & \qquad \text{(similar to proof for the second part $Y_2$)}\\
    & \leq \int_0^t \int_0^t \Big\{ (\frac{\beta^2}{r}\|\Fs_\tau\|_2+ |\eta|\|\Fz_\tau\|_2)(\frac{\beta^2}{r}\|\Fs_\theta\|_2+ |\eta|\|\Fz_\theta\|_2) \\
    & \quad \cdot \Big[(e^{|\eta|c (t-\tau)}-1)(e^{|\eta|c (t-\theta)}-1) \mu^2 E_N^2  + (e^{|\eta|c (t-\tau)}) \mu^2 E_N^2 \\
    & \quad + (e^{|\eta|c (t-\theta)}) \mu^2 E_N^2  + \mu^2 E_N^2\Big]\Big\}d\tau d\theta\\
    & =\mu^2 E_N^2\int_0^t \int_0^t \Big\{   (\frac{\beta^2}{r}\|\Fs_\tau\|_2+ |\eta|\|\Fz_\tau\|_2)(\frac{\beta^2}{r}\|\Fs_\theta\|_2+ |\eta|\|\Fz_\theta\|_2) \\
    & \quad \cdot \Big[(e^{|\eta|c (t-\tau)}-1)(e^{|\eta|c (t-\theta)}-1)  + (e^{|\eta|c (t-\tau)})  + (e^{|\eta|c (t-\theta)})  + 1\Big]\Big\}d\tau d\theta\\
    \end{aligned}
\end{equation}
The above analysis for $Y_1,Y_2,Y_3,Y_4$ and $Y_5$ implies that for all $i\in\{1,...,N\}$
\begin{equation}
    \left[\mathds{E}\big[\Delta_t^N \Delta_t^N\strut^\TRANS \big]\right]_{ii}= \max\Big\{O\Big(\frac{1}{N}\Big), O(E_N^2)\Big\}.
\end{equation}
Since by definition 
$
    I_2 = \int_0^T \left[\mathds{E}\big[\Delta_t^N \Delta_t^N\strut^\TRANS \big]\right]_{ii} dt,
$
we obtain
\begin{equation}\label{eq:I2-bigO-randinit}
   | I_2| = \max\Big\{O\Big(\frac{1}{N}\Big), O(E_N^2)\Big\}, \quad \forall i \in \{1,..., N\}.
\end{equation}



\subsubsection*{ii) Estimate for $I_3$}
Setting $\delta_t^{Ni} \triangleq \bar \Fz^i_t- x_t^{oi}$ yields 
\begin{equation*}
  \begin{aligned}
      \dot \delta_t^{Ni} 
   &= (\alpha - \frac{\beta^2}{r} \pi_t)\delta_t^{Ni} 
   +\frac{\beta^2}{r} \frac{1}{\mu(P_i)}
\langle\mathds{1}_{P_i}, \FA\Fs_t -  \Fs_t \rangle + \eta \frac{1}{\mu(P_i)} \langle  \mathds{1}_{P_i},  \FA \Fz_t -  \Fz^{o\FN}_t\rangle \\
   &= (\alpha - \frac{\beta^2}{r} \pi_t)\delta_t^{Ni} 
   +\frac{\beta^2}{r} \frac{1}{\mu(P_i)}
\langle(\FA-\BI)\mathds{1}_{P_i}, \Fs_t \rangle\\
& \quad + \eta \frac{1}{\mu(P_i)} \langle  (\FA- \BI)\mathds{1}_{P_i},\Fz_t\rangle + \eta \frac{1}{\mu(P_i)} \langle  \mathds{1}_{P_i},  \bar \Fz_t -  \Fz^{o\FN}_t\rangle\\
 &= (\alpha - \frac{\beta^2}{r} \pi_t)\delta_t^{Ni} 
   +\frac{\beta^2}{r} \frac{1}{\mu(P_i)}
\langle(\FA-\BI)\mathds{1}_{P_i}, \Fs_t \rangle\\
& \quad + \eta \frac{1}{\mu(P_i)} \langle  (\FA- \BI)\mathds{1}_{P_i},\Fz_t\rangle + \eta \Delta_t^{Ni}.
  \end{aligned}
\end{equation*}
Therefore, 
\begin{equation*}
  \begin{aligned}
       \delta_t^{Ni} =& 
       \Gamma(t,0)\delta_0^{Ni} + \int_0^t \Gamma(t,\tau) \frac{\beta^2}{r} \frac{1}{\mu(P_i)} \langle(\FA-\BI)\mathds{1}_{P_i}, \Fs_\tau\rangle d\tau \\
       &+ \int_0^t \Gamma(t,\tau) \eta \frac{1}{\mu(P_i)} \langle (\FA-\BI) \mathds{1}_{P_i},\Fz_\tau\rangle  d\tau + \int_0^t \Gamma(t,\tau) \eta \Delta_\tau^{Ni}  d\tau \\
         \leq& \Gamma(t,0)\|\FA-\BI\|_\text{op}\|\Fx_0^\FN\|_2 + \int_0^t \Gamma(t,\tau) \|(\FA-\BI)\|_\text{op} \Big(\frac{\beta^2}{r}\|\Fs_\tau\|_2 +\eta \|\Fz_\tau\|_2 \Big)  d\tau \\
       &~~+ \int_0^t \Gamma(t,\tau) \eta \Delta_\tau^{Ni}  d\tau \\
    \end{aligned}
\end{equation*}
with the initial condition
$
\delta_0^{Ni} = \bar \Fz^i_0- x_0^{oi}  =  \frac{1}{\mu(P_i)}  \mathds{1}_{P_i}^\TRANS\FA \mathds{1}\mu - x_0^i
$. 
Since $ \mathds{E}\delta_0^{Ni} = \bar \Fz_0^i - \mu $ and $\mathds{E}(\delta_0^{Ni})^2= (\bar \Fz^i_0-\mu)^2+ \sigma^2$,
they are uniformly bounded in $N$.
%
Then 
\begin{equation} \label{eq:deltatNi}
\begin{aligned}
        \mathds{E}[\delta_t^{Ni}]^2 
     =&  \mathds{E}\Big(\Gamma(t,0)\delta_0^{Ni} + \int_0^t \Gamma(t,\tau)(\frac{\beta^2}{r}D^i_\Fs(\tau)+ \eta D^i_\Fz(\tau))+ \eta \int_0^t\Gamma(t,\tau) \Delta_\tau^{Ni} d\tau \Big)^2.
\end{aligned}
\end{equation}
Recall from \eqref{eq:Delta-Expectation-Ni} that 
$
    |\mathds{E}\Delta_0^{Ni}| \leq |\mu| E_N
$.
Therefore,
\begin{equation}\label{eq:Delta-i-randinit}
    \mathds{E}\Delta_t^{Ni} \Delta_t^{Ni}\strut^\TRANS = \max\Big\{O\Big(\frac{1}{N}\Big), O(E_N^2)\Big\}, \quad \forall i \in \{1,...,N\}.
\end{equation}

By expanding and evaluating all the terms in \eqref{eq:deltatNi}, we obtain
  $\mathds{E}(\delta_t^{Ni})^2 = \max\left\{O (1), O(E_N^2)\right\}$.   Therefore, 
  we obtain
\begin{equation}\label{eq:I-3-randinit}
 | I_3| = \max\left\{O\Big(\frac{1}{\sqrt{N}} \Big), O\Big(E_N \Big), O\Big(E_N^2 \Big)\right\} . 
\end{equation}
\subsubsection*{iii) Estimates for $I_2^\prime$ and $I_3^\prime$}
Next we obtain the rate of convergence for $I_2^\prime$ and $I_3^\prime$. 
Under Assumption \ref{ass:a_ii=0}, 
\begin{equation}
\begin{aligned}
  I_2^\prime & =\mathds{E}  \int_0^T \Big[(\bar\Fz_t^i-\frac1N\sum_{j=1}^{N}a_{ij}x^j_t)\Big]^2_{({u}^{i},u^{-o i})}dt\\
  & = \mathds{E} \int_0^T \Big[(\bar\Fz_t^i-\frac1N\sum_{j=1}^{N}a_{ij}x^j_t)\Big]^2_{({u}^{oi},u^{-o i})}dt.
\end{aligned}
\end{equation}
Hence for fixed $T>0$,
\begin{equation}\label{eq:I-2prime-randinit}
  |I_2^\prime| = \max\Big\{O\Big(\frac{1}{N}\Big), O(E_N^2)\Big\}, \quad \forall i \in \{1,...,N\}. 
\end{equation}
Similarly to \eqref{eq:I-3}, we obtain 
\begin{equation}\label{eq:I-3prime-randinit}
|I_3^\prime| = \max\left\{O\Big(\frac{1}{\sqrt{N}} \Big), O\Big(E_N \Big), O\Big(E_N^2 \Big)\right\} .
\end{equation} \\
%

\end{proof}
\subsubsection*{Lemmas for  Theorem \ref{thm:stochasticinit-LQ-GFG}} \label{sec:Lemmas} ~\\~

 Let $x^{oi}, i\in\{1,...,N\}$, denote the state trajectory of agent $i$ when all agents are following Strategy \ref{str:randinit}.  The initial states of all agents are independent and identically distributed $\text{N}(\mu, \sigma^2)$. 
Assume the following information is given to an arbitrary agent indexed by $i$:
\begin{itemize}
	\item the reference trajectory $\bar \Fz^i_t\triangleq 
   \frac{1}{\mu(P_i)}\int_{P_i} \Fz^\gamma_t d\gamma$ for all $t\in [0,T]$ with $\Fz$ as the graphon filed of the corresponding limit graphon field game problem satisfying \eqref{eq:fixed-point} where the initial condition is replace by $\Fz_0 = \mu \mathds{1}$
	\item the dynamic offset $z_t^{oi} \triangleq \frac{1}{N}\sum_{j\in \mathcal{N}_i} a_{ij}x_t^{oj}$ for all $t\in [0,T]$ where $\mathcal{N}_i$ denotes the set of neighours for agent $i$ excluding itself.
\end{itemize}
Then consider the following linear quadratic tracking problem for agent $i$:
\begin{equation}\label{eq:complete-info-lqt-randinit}
\begin{aligned}
    &\dot x_t^{i} = \alpha x_t^{i} + \beta u_t^{i}+\eta \frac{1}{N}\sum_{j\in \mathcal{N}_i} a_{ij}x_t^{oj}\\
  &  J(u^{i}, \bar \Fz^i ) = \frac12 \mathds{E} \int_0^T\left[ (x^{i}_t-\bar \Fz_t^i)^2 + r(u_t^{i})^2 \right]dt .
\end{aligned}
\end{equation}
where the random initial condition is distributed $\text{N}(\mu, \sigma^2)$. 


\begin{lemma}\label{lem:opt-mfg-cost-comp-randinit}
 Under 
  Assumptions \ref{ass:limit-assumptions}(a), \ref{ass:spectral-graphon}, \ref{ass:Riccati-Sol-Existence}
and \ref{ass:a_ii=0}, 
 the following estimate  for the costs in problem \eqref{eq:complete-info-lqt-randinit} holds when $u^{oi}$ is generated based on Strategy \ref{str:randinit}:
%
%
%
  \begin{equation}\label{eq:cost-opt-track-ave-randinit}
  J(u^{oi}, \bar \Fz^i) -  \inf_{u^i\in \mathcal{U}}J(u^{i}, \bar \Fz^i)  = O \Big(E_N^2\Big),
\end{equation}
where  $\mathcal{U} = L^2([0,T];\BR)$ and $E_N\triangleq \max_{1\leq i\leq N} \frac{1}{\mu(P_i)} \left\|  (\FA-\FA^\FN)\mathds{1}_{P_i}\right\|_2.$
\end{lemma}
\begin{proof}
The proof is similar to that of Lemma \ref{lem:opt-mfg-cost-comp}. 
The optimal control law for the problem in \eqref{eq:cost-opt-track-ave-randinit} is given by
  \begin{align}
      u_t^{*i} &= -\frac{\beta}{r} \pi_t x_t^{*i}+ \frac{\beta}{r} s_t^{*i},\\
      -\dot{\pi}_t &= 2\alpha \pi_t -\frac{\beta^2}{r}\pi^2_t + 1, \qquad  \pi_T = 0, \\ 
      -\dot{s}_t^{*i} & = \Big(\alpha - \frac{\beta^2}{r} \pi_t \Big) s_t^{*i} + \bar\Fz_t^i-\eta \pi_t z_t^{oi},\quad
        s_T^{*i} = 0.\label{eq:s-star-randinit}
  \end{align}
  Then the corresponding dynamics and cost are given by 
\begin{align}
    &\dot x_t^{*i} = \alpha x_t^{*i} + \beta u_t^{*i}+\eta \frac{1}{N}\sum_{j\in \mathcal{N}_i} a_{ij}x_t^{oj}\\
  &  J(u^{*i}, \bar \Fz^i ) =  \frac12 \mathds{E}\int_0^T\left[ (x^{*i}_t-\bar \Fz_t^i)^2 + r(u_t^{*i})^2 \right]dt 
\end{align}
On the other hand, following Strategy \ref{str:randinit}, the response is as follows: 
\begin{equation}
\begin{aligned}
  {u}^{oi}_t  & =-\frac{\beta}{r} \pi_t x_t^{oi}+ \frac{\beta}{r} \bar\Fs_t^i 
\end{aligned}
\end{equation}
where $\bar \Fs^i$ is defined according to \eqref{eq:rdinit-Nash-Control-Case12}. Note that under 
  Assumptions \ref{ass:limit-assumptions}(a), \ref{ass:spectral-graphon} and \ref{ass:Riccati-Sol-Existence},  $\bar \Fs^i$  always exists. 
The corresponding dynamics and cost are given by
\begin{align}
    &\dot x_t^{oi} = \alpha x_t^{oi} + \beta u_t^{oi}+\eta \frac{1}{N}\sum_{j\in \mathcal{N}_i} a_{ij}x_t^{oj},\\
  &  J(u^{oi}, \bar \Fz^i ) = \frac12 \mathds{E}\int_0^T\left[ (x^{oi}_t-\bar \Fz_t^i)^2 + r(u_t^{oi})^2 \right]dt 
\end{align}
Based on \eqref{eq:fixed-point} and the definitions of $\bar \Fs^i$ and $\bar \Fz^i$, we obtain
\begin{align}
    \dot{\bar \Fs}_t^i&=-\Big(\alpha-\frac{\beta^2}{r}\pi_t\Big)\bar \Fs_t^i - (1-\eta \pi_t)\bar \Fz_t^i,  \quad  \bar\Fs^i_T = 0.
\end{align}
This together with \eqref{eq:s-star-randinit} yields
\begin{equation}
  \frac{d(\bar \Fs_t^i - s_t^{*i})}{dt} = -\Big(\alpha-\frac{\beta^2}{r}\pi_t\Big) (\bar \Fs_t^i - s_t^{*i}) +\eta \pi_t(\bar \Fz_t^i- z_t^{oi}).
\end{equation}
Let $\Delta_t^{Ni} \triangleq \bar \Fz_t^i - z_t^{oi}.$ Then for all $i\in\{1,...,N\}$, 
\begin{equation}\label{eq:s-diff-dyn}
  (\bar \Fs_t^i - s_t^{*i}) = \int_T^t \Gamma(t,\tau)\eta \pi_\tau \Delta_\tau^{Ni} d\tau.
\end{equation}
Under Assumption \ref{ass:a_ii=0}, $\sum_{j\in \mathcal{N}_i}a_{ij}x_t^j = \sum_{j=1}^N a_{ij} x_t^j$
and hence the result in \eqref{eq:Delta-i-randinit} applies here.
That is, for any $t\in[0,T]$,
\begin{equation}\label{eq:Delta-i-randinit-in-lemma}
    \mathds{E}(\Delta_t^{Ni})^2 = \max\Big\{O\Big(\frac{1}{N}\Big), O(E_N^2)\Big\}.
\end{equation}
This together with \eqref{eq:s-diff-dyn} implies
\begin{equation}\label{eq:s-diff-rinit}
  \mathds{E}(\bar \Fs_t^i - s_t^{*i})^2 = \max\Big\{O\Big(\frac{1}{N}\Big), O(E_N^2)\Big\}.
\end{equation}
Comparing closed-loop dynamics under the two different control laws yields
\begin{equation*}
\begin{aligned}
    \frac{d(x_t^{oi} - x_t^{*i})}{dt} & = \Big(\alpha -\frac{\beta^2}{r} \pi_t \Big)(x_t^{oi}- x_t^{*i}) + \frac{\beta^2}{r}(\bar \Fs_t^i - s^{*i}_t),~~
    x_0^{oi} - x_0^{*i}  = 0.
\end{aligned}
\end{equation*}
The difference is explicitly given by
\begin{equation}\label{eq:x-diff-soln-randinit}
\begin{aligned}
 x_t^{oi} - x_t^{*i}  &= \int_0^t \Gamma(t,\tau)\frac{\beta^2}{r}(\bar \Fs_\tau^i - s^{*i}_\tau)d\tau\\
  & = \frac{\eta\beta^2}{r}\int_0^t \Gamma(t,\tau)\left[\int_T^\tau \Gamma(\tau,q)\eta \pi_q \Delta_q^{Ni} dq\right]d\tau.
\end{aligned}
\end{equation}
Under Assumption \ref{ass:a_ii=0}, $\sum_{j\in \mathcal{N}_i}a_{ij}x_t^j = \sum_{j=1}^N a_{ij} x_t^j$ and the result from \eqref{eq:expectation-Delta-randinit} applies here, i.e.,
\[
\forall t\in[0,T], \quad \|\mathds{E}\Delta_t^N \|_\infty =\max_i| \mathds{E}\Delta_t^{Ni}|= O(E_N).
\]
This together with \eqref{eq:x-diff-soln-randinit} and  \eqref{eq:s-diff-dyn} implies
\begin{equation} \label{eq:RILStateEXP}
	|\mathds{E}(x_t^{oi}-x_t^{*i})| = O(E_N) \quad\text{and}\quad  |\mathds{E}(\bar \Fs_t^i - s_t^{*i})| = O(E_N).
\end{equation}
Therefore, by the construction of the two control laws,  we obtain
\begin{equation} \label{eq:RILControlEXP}
    |\mathds{E}(u_t^{oi} - u_t^{*i})| = O(E_N).
\end{equation}
%
%
%
 Equations \eqref{eq:Delta-i-randinit-in-lemma}, \eqref{eq:s-diff-rinit} and \eqref{eq:x-diff-soln-randinit} imply that
\begin{equation}\label{eq:RILStateVar}
  \mathds{E}(x_t^{oi} - x_t^{*i})^2= \max\Big\{O\Big(\frac{1}{N}\Big), O(E_N^2)\Big\}. 
\end{equation}
This together with \eqref{eq:s-diff-rinit} implies
\begin{equation} \label{eq:RILControlVar}
  \mathds{E}(u_t^{oi} - u_t^{*i})^2= \max\Big\{O\Big(\frac{1}{N}\Big), O(E_N^2)\Big\} 
\end{equation}
by the construction of the two control laws.
We observe that
\begin{equation}\label{eq:RILcost-seperation}
\begin{aligned}
      J(u^{*i}, \bar \Fz^i ) 
          &=\frac12  \mathds{E}\int_0^T\left[ (x^{*i}_t-x_t^{oi}+ x_t^{oi}-\bar \Fz_t^i)^2 + r(u_t^{*i}-u_t^{oi}+u_t^{oi})^2 \right]dt \\
        &= J(u^{oi},\bar\Fz^i)
         + \frac12 \mathds{E}\int_0^T\left[ (x^{*i}_t-x_t^{oi})^2 + r(u_t^{*i}-u_t^{oi})^2 \right]dt \\
         &\quad +  \int_0^T\left[ \mathds{E}[x^{*i}_t-x_t^{oi}]( x_t^{oi}-\bar \Fz_t^i) + r\mathds{E}[u_t^{*i}-u_t^{oi}](u_t^{oi}) \right]dt.
\end{aligned}
\end{equation}
Therefore, based on  \eqref{eq:RILStateEXP}, \eqref{eq:RILControlEXP}, \eqref{eq:RILStateVar}, \eqref{eq:RILControlVar} and \eqref{eq:RILcost-seperation},  we obtain 
\begin{equation}
  \left|J(u^{oi}, \bar \Fz^i) - J(u^{*i}, \bar \Fz^i) \right| = \max\Big\{O\Big(\frac{1}{N}\Big),O(E_N), O(E_N^2)\Big\}
\end{equation}
that is, \eqref{eq:cost-opt-track-ave-randinit} holds.

\end{proof}

\begin{lemma}\label{eq:lem-expA-Bij}
If $|a_{ik}|\leq c_a$ and $|b_{k j}|\leq c_b $ for all $i,k \in \{1,2\ldots,N\}$, $j \in \{1,\ldots,M\}$ with $N, M\in \{1,2,\ldots\}$,  then the following inequality 
   $ \Big|\big[e^{\eta \frac{A}{N}} B \big]_{ij}\Big| \leq c_b e^{c_a|\eta|}$
for all $i \in \{1,\ldots,N\}$, $j\in \{1,\ldots, M\}$, 
where $A =[a_{ik}]$ and $B=[b_{k j}]$.   
\end{lemma}
\begin{proof}
    \begin{equation}
    \begin{aligned}
   \Big|\big[e^{\eta \frac{A}{N}} B \big]_{ij}\Big| & = \Big|\big[(e^{\eta \frac{A}{N}}-I) B + B\big]_{ij}\Big|  \\
    & \leq \Big|\big[(\eta \frac{A}{N}+ \eta^2 \frac{A^2}{N^2}\frac1{2!}+\ldots) B\big]_{ij}\Big|+ c_b\\
    & = \Big|\sum_{k=1}^N \Big[\eta \frac{A}{N}+ \eta^2 \frac{A^2}{N^2}\frac1{2!}+\ldots\Big]_{ik} b_{kj} \Big| + c_b\\
    & \leq c_b \sum_{k=1}^N \Big|\big[\eta \frac{A}{N}+ \eta^2 \frac{A^2}{N^2}\frac1{2!}+\ldots\big]_{ik}\Big| + c_b\\
    & \leq c_b N \left(|\eta| \frac{c_a}{N}+ |\eta|^2 \frac{c_a^2}{N}\frac1{2!}+\ldots\right) + c_b\\
    &= c_b (e^{c_a|\eta|}-1) + c_b = c_b e^{c_a|\eta|}.
    \end{aligned}
\end{equation}
\end{proof}
Applications of Lemma \ref{eq:lem-expA-Bij} yield the following results. 
\begin{lemma}\label{lem:Ex-vector}
If $|a_{ij}|\leq c_a$ and $|v_i| \leq c_v$  for all $i,j \in  \{1,\ldots,N\}$, where $v\in \BR^N$, then 
        $|[e^{\eta \frac{A}{N}}v]_i|  \leq c_v e^{|\eta| c_a}$  for all $i,j \in  \{1,\ldots,N\}$,
    where $A =[a_{ij}]$.
\end{lemma}

\begin{lemma}\label{eq:point-bound-expA}
        If $|a_{ij}|\leq c$ for all $i,j \in \{1,2\ldots,N\}$, 
        then
    $
        \left|\left[e^{\eta \frac{A}{N}} \frac{A}{N} \right]_{ij}\right| \leq \frac{c}{N}e^{c|\eta|}\textbf{}
   $ for all $i,j \in \{1,2\ldots,N\}$,
    where $A =[a_{ij}]$.
\end{lemma}
\begin{lemma} \label{eq:lem-expA-Wij} Let $W=[w_{ij}], A= [a_{ij}] \in \BR^{N\times N}$. 
If  $|w_{ij}|\leq c_w$ and $|a_{ij}|\leq c_a$ for all $i,j \in \{1,\ldots,N\}$, then the following inequalities hold:
    \begin{equation}
        \begin{aligned}
        &\left|\left[\big(e^{\eta \frac{A}{N}}-I\big) W\right]_{ij}\right| \leq (e^{|\eta|c_a}-1) c_w,\\
        &   \left|\left[\big(e^{\eta \frac{A}{N}}-I\big) W\big(e^{\eta \frac{A}{N}}-I\big)^\TRANS\right]_{ij}\right| \leq (e^{|\eta|c_a}-1)^2 c_w.
        \end{aligned}
    \end{equation}
\end{lemma}
%
\end{document}